\documentclass[10pt]{article}
\usepackage{arxiv}

\begin{document}
	
\title{Fast and Accurate Triangle Counting in Graph Streams Using Predictions \\
}

\hypersetup{
	pdftitle={Fast and Accurate Triangle Counting in Graph Streams Using Prediction},
	pdfauthor={Cristian Boldrin, Fabio Vandin},
	pdfkeywords={Data Mining, Big Data, Graphs, Triangle Counting, Streaming},
}

\author{
	Cristian Boldrin \\ 
	\small{Dept. of Information Engineering} \\
	\small{University of Padova, Italy} \\
	\small{\texttt{boldrincri@dei.unipd.it}}  
\and
	Fabio Vandin \\
	\small{Dept. of Information Engineering} \\
	\small{University of Padova, Italy} \\
	\small{\texttt{fabio.vandin@unipd.it}}
}

\date{September, 2024}

\maketitle

\begin{abstract}
	In this work, we present the first efficient and practical algorithm for estimating the number of triangles in a graph stream using \emph{predictions}. Our algorithm combines waiting room sampling and reservoir sampling with a \emph{predictor} for the \emph{heaviness} of edges, that is, the number of triangles in which an edge is involved. As a result, our algorithm is fast, provides guarantees on the amount of memory used, and exploits the additional information provided by the predictor to produce highly accurate estimates.  We also propose a simple and domain-independent predictor, based on the degree of nodes, that can be easily computed with one pass on a stream of edges when the stream is available beforehand.
	
	Our analytical results show that, when the  predictor provides useful information on the heaviness of edges, it leads to estimates with reduced variance compared to the state-of-the-art, even when the predictions are far from perfect. Our experimental results show that, when analyzing a single graph stream, our algorithm is faster than the state-of-the-art for a given memory budget, while providing significantly more accurate estimates. Even more interestingly, when sequences of hundreds of graph streams are analyzed, our algorithm significantly outperforms the state-of-the-art using our simple degree-based  predictor built by analyzing only the first graph of the sequence.
\end{abstract}



\section{Introduction}
\label{sec:intro}
Counting the number of triangles is a fundamental primitive in the analysis of graphs, with applications ranging from community detection~\cite{berry2011tolerating} and anomaly detection~\cite{becchetti2010efficient, lim2015mascot} to molecular biology~\cite{milo2002network}. In most applications the exact computation of the number of triangles is unfeasible, due to the massive size of the data. For this reason, one has often to resort to efficient algorithms that provide high-quality approximations, that can be used in place of exact values in subsequent analyses.  One of the main requirements for such algorithms is that they use a limited amount of memory, for example storing only a small fraction of the edges of the graph to not exceed a given memory budget.

In some applications the graphs of interest are not only massive, but they are also observed as a stream of edges that occur in arbitrary order. In such situations one is interested in keeping \emph{high-quality} approximations \emph{at every time} during the stream. Even when the data is not observed as a stream, analyzing a graph through one (or few) passes on its edges leads to efficient algorithms.

A lot of work has been done on approximate triangle counting in streams, often using sampling (see Section~\ref{sec:related}). For example,~De Stefani et al.~\cite{stefani2017triest} introduced the use of \emph{reservoir sampling} to keep the memory bounded \emph{with certainty} and, at the same time, make better use of the  available memory. More recently~Shin et al.\cite{shin2017wrs} introduced the notion of \emph{waiting room sampling}, which allows to take advantage of temporal localities of triangles observed in most applications by keeping the most recent edges in the stream in a waiting room. While such algorithms provide accurate and efficient solutions for counting triangles in data streams, there is still room for improvement. For example, such algorithms cannot adapt to the data distribution that is specific to each application, and that may be captured by using a \emph{predictor}, such as a machine learning model, for example, providing additional and relevant information on the graph~\cite{mitzenmacher2022algorithms}.

\textit{Contributions.} We introduce  \algname\ (\emph{T}riangles c\emph{O}u\emph{N}ting with pred\emph{IC}tions), a suite of novel efficient algorithms for fast and accurate triangle counting in graph streams. \algname\ can be instantiated for both insertion-only and fully dynamic streams with both edge insertions and deletions. For ease of presentation and due to space constraints, we focus on the insertion-only variant of \algname; the variant of \algname\ for fully dynamic streams is described in Appendix~\ref{sec:additionalalgos}. Our algorithms use uniform sampling schemes to store a small sample of the stream while making sure to not exceed the allowed memory budget.
\begin{myitemize}
	\item  \algname\ is the first \emph{fast} and \emph{accurate} algorithm for approximating the number of triangles in graph streams using \emph{predictions}. \algname\ combines such predictions with reservoir sampling and waiting room sampling to provide high-quality estimates of both the \emph{global} number of triangles and the \emph{local} number of triangles incident to each vertex. \algname\ can use any predictor providing some measure of \emph{heaviness} for edges, defined as the number of triangles an edge is involved in. By using a predictor, \algname\ focuses on the most relevant edges for triangle counting, and it adapts to the data distribution specific to the application. Our analysis shows that the predictor leads to improved estimates when it provides useful information on heavy edges, even if its predictions are far from perfect.
	\item We propose a very simple and application independent predictor, based on the degree of nodes, to be used in \algname. Such predictor can be learned with one pass of the stream and can be easily stored. This is in contrast with previous work exploring the use of \emph{heaviness} predictors for triangle counting in graph streams~\cite{chen2022triangle}, and makes \algname\ the first \emph{practical} approach for triangle counting in data streams with predictions.
	\item We conduct an extensive experimental evaluation on very large graph streams and on sequences of graph streams. The results show that, when analyzing a single graph stream, \algname\ is faster than the state-of-the-art for a given memory budget, while providing significantly more accurate estimates. The improvement is even more significant for sequences of hundreds of graph streams, where \algname\ substantially outperforms the state-of-the-art using our simple degree-based predictor built analyzing only the \emph{first} graph stream.
\end{myitemize}

\section{Related Works}
\label{sec:related}

The problem of counting global and local triangles in a graph has been extensively studied in the last decades~\cite{tsourakakis2009doulion,tsourakakis2011triangle,kane2012counting,manjunath2011approximate,jayaram2021optimal,mcgregor2016better}, in many different settings~\cite{pagh2012colorful,park2014mapreduce,park2013efficient,park2016pte,suri2011counting}. Due to space constraints, we discuss here the works mostly related to ours, focusing on sampling approaches to estimate triangle counts in graph streams, and refer to the surveys \cite{latapy2008main, al2018triangle} for a more in-depth presentation of other approaches. Sampling approaches fall into two main categories: \textit{fixed memory}, where edges are sampled without exceeding a given memory budget, and \textit{fixed probability}, where edges are sampled with a given fixed probability. While the two approaches are related (by appropriately setting the sampling probability in fixed probability approaches), the latter does not guarantee that the memory budget is not exceeded. Since our focus is on a fixed memory budget, we discuss only such kind of approaches.

Buriol et al.~\cite{buriol2006counting} presented a suite of sampling algorithms, resulting in $(1 + \epsilon)$-approximation of the global triangles. Pavan et al.~\cite{pavan2013counting} introduced a one-pass stream efficient  approach that runs different estimators to approximate the count of triangles, providing guarantees on the memory used. Jha et al.~\cite{jha2013space} provided an algorithm for approximating the global number of triangles with a single pass through a graph stream using the birthday paradox.
De Stefani et al.~\cite{stefani2017triest} proposed \textit{Trièst}, a suite of algorithms for counting the global and local number of triangles in fully dynamic streams using reservoir sampling~\cite{vitter1985random} (see Section~\ref{sec:reservoirs}) to keep edges within a fixed memory budget, and random pairing~\cite{gemulla2008maintaining} to maintain a bounded-size uniform random sample of the edges in a fully dynamic stream. Shin et al.~\cite{shin2017wrs} and Lee et al.~\cite{lee2020temporal} introduced the idea of waiting room sampling (see Section~\ref{sec:wrs}), which consists of storing the most recent edges, based on the empirical observation that triangles have edges spanning a short interval in the stream. Shin et al.~\cite{shin2018think, shin2020fast} presented 
$\textit{ThinkD}$, a suite of algorithms for handling fully dynamic graph streams with random pairing, that builds on the scheme introduced by~\cite{stefani2017triest} resulting in highly accurate estimates.


The use of predictions about the input has been recently formalized in the \textit{algorithms with predictions} framework (see the survey by Mitzenmacher and Vassilvitskii~\cite{mitzenmacher2022algorithms}), and several algorithms with predictions for well-studied problems have been proposed~\cite{mitzenmacher2018model,khalil2017learning,hsu2019learning}. The problem of counting \emph{global} triangles in a \emph{insertion-only} graph stream has been studied in such framework by Chen et al.~\cite{chen2022triangle} lately. They proposed one-pass streaming algorithms for estimating the number of global triangles and four cycles using an oracle that provides predictions about the heaviness of edges. In this regard, the algorithm from~\cite{chen2022triangle} is similar to our algorithm for insertion-only graph streams, but there are three crucial differences. First, we also propose a variant of our algorithm for fully dynamic graph streams, while~\cite{chen2022triangle} limits to insertion-only graph streams. Second, the algorithm they propose and analyze uses a fixed probability approach, with no guarantees to respect the allowed memory budget and requiring knowledge of the number of edges in the stream to compute the sampling probability. Third, they assume that the oracle is available to the algorithm, but do not propose practical and efficiently computable oracles, that is instead one of our main contributions. As a result, our algorithm \algname\ is the first efficient and practical approach that uses predictions for estimating the number of triangles in (fully dynamic) graph streams, as shown by our experimental evaluation (see Section~\ref{sec:experiments}).





\section{Preliminaries}
We now introduce the main notions used in this paper. 
We consider an undirected graph $G = (V, E)$ with no self-loops and no multiple edges, where $V$ and $E$ are the set of nodes and the set of edges, respectively, with $|V| = n$ and $|E| = m$. Edges are observed in arbitrary order through the graph stream $\Sigma = \{e^{(1)}, ..., e^{(m)} \}$. An edge $e^{(t)} = \{u, v\} \in E$ is an unordered pair of nodes. Note that $\Sigma$ is an \emph{insertion-only} stream: edges can only be added and not deleted. Our algorithm \algname\ can be adapted to fully dynamic streams with arbitrary edge insertions and deletions, but we present only the variant for insertion-only streams due to space constraints and for simplicity. The description of our algorithm for fully dynamic streams is in Appendix~\ref{sec:additionalalgos}.

For $t = 1, ..., m$, we denote with $G^{(t)} = (V, E^{(t)})$, where $E^{(t)} =  \{e^{(1)},\dots,e^{(t)} \}$, the graph up to time $t$; note that $G^{(m)}= G$.  We say that $\Delta = \{u,v,w\}$ is a \emph{triangle} in $G^{(t)}$ if edges $\{u,v\}, \{u,w\}$, and $\{v,w\}$ all appear in $E^{(t)}$.  We also say that a triangle $\Delta$ in $G^{(t)}$ is incident to a vertex $v$ if $v \in \Delta$.

For every $t=1,2,\dots$, we are interested in approximating the \emph{global} number $T^{(t)}$  of  triangles in $G^{(t)}$ as well as the \emph{local} number $T_u^{(t)}$ of triangles incident to $u$ in $G^{(t)}$, for every $u \in V$. In our work, we make the following assumptions: no exact information about the input stream (e.g., the true number of nodes, the true number of edges, the true number of triangles) is available to the algorithm; we can store at most $k$ edges in memory ($k$ is part of the input). We are also going to assume that the edge stream can be only accessed once (i.e., we consider one-pass streaming algorithms) to count triangles. However, we also  describe a simple but effective predictor that can be obtained with a fast additional pass on the stream $\Sigma$.

\subsection{Waiting Room Sampling}
\label{sec:wrs}
Our algorithm \algname\ uses waiting room sampling, proposed by Shin\soutc{et al.}~\cite{shin2017wrs}, that allows to exploit \emph{temporal localities} observed in most real-world datasets. Temporal localities represent the tendency that future edges are more likely to form triangles with recent edges rather than with older edges in real graph streams. To exploit such temporal localities in triangle counting, Shin\soutc{et al.}~\cite{shin2017wrs} propose to always store the most recent edges of the stream. More in detail, in waiting room sampling with memory budget $k$, for some constant $\alpha\in (0,1)$, a portion $k  \alpha$ of the memory, called \emph{waiting room} $\mathcal{W}$, is reserved to keep the $k  \alpha$ most recent edges from the stream $\Sigma$.

\subsection{Predictor}
\label{sec:pred}
Our algorithm makes use of a \emph{predictor}. We consider predictors that provide some information about the \emph{heaviness} of edges, that is, the number of triangles in which edges are involved. More formally, for an edge $e=\{u,v\}$, let $\Delta(e)$ be the number of triangles containing both $u$ and $v$ (i.e., $e$ is an edge of the triangle). We define a predictor $O_H$ for the heaviness of edges as a function $O_H: E \rightarrow \mathbb{R}^+$, where $O_H(e)$ is a measure \emph{related} to $\Delta(e)$: our algorithm uses $O_H$ only to \emph{compare} the heaviness of edges, so $O_H(e)$ does not need to be a prediction for the actual value of $\Delta(e)$. In particular, for every time $t$, our algorithm \algname\ uses $O_H$ to keep a fixed size set with the heaviest edges observed up to time $t$.

The intuition for using such predictions is that in most cases triangles are not distributed equally across edges, thus $O_H$ allows to focus on edges that most contribute to the number of triangles. The impact of $O_H$ then depends on how much the edges it predicts as \emph{heavy} actually contribute to triangle counting (in relation to the other edges).

Note that a heaviness predictor is targeting global counting of triangles, in contrast to the local number of triangles incident to \emph{each} vertex. However, a heaviness predictor is useful also for vertices with \emph{high} local triangles counts, which are often the ones of interest when analyzing local counts.


\subsection{Reservoir Sampling}
\label{sec:reservoirs}
The last component of our algorithm \algname\ is the  sampling of edges through \emph{reservoir sampling}~\cite{vitter1985random}. In particular, at time $t$,  \algname\ stores a set of $k' < k$ edges sampled uniformly at random among the \emph{light edges} of $G^{(t)}$, that are edges of $G^{(t)}$ that are neither in the waiting room $\mathcal{W}$ nor kept in the set of (predicted) heavy edges (see Section~\ref{sec:alg}). More in detail, let $L^{(t)}$ be the set of light edges in $G^{(t)}$. At time $t$, reservoir sampling maintains a sample $\mathcal{S}_L$ of at most $k'$ edges as follows: let $e^{(t)}$ be the edge observed at time $t$; if $|L^{(t)}| < k'$, then $e^{(t)}$ is added to $\mathcal{S}_L$; otherwise, with probability  $k'/|L^{(t)}| $ remove an edge chosen uniformly at random from $\mathcal{S}_L$ and add $e^{(t)}$ to $\mathcal{S}_L$ (with probability $1-k'/|L^{(t)}|$, $\mathcal{S}_L$ does not change).

Let $\mathcal{S}_L^{(t)}$  be the set of edges in $\mathcal{S}_L$ at the end of time step~$t$. Reservoir sampling guarantees~\cite{vitter1985random} that for every time step $t$ with $|L^{(t)}| \ge k'$, if we let $A$ be any subset of $L^{(t)}$ of size $|A|~=~k'$, then $ \mathbb{P}\left[\mathcal{S}_L^{(t)}= A\right]~=~\frac{1}{\binom{|L^{(t)}|}{k'}}$,
that is, $\mathcal{S}_L^{(t)}$ is a uniform sample of size $k'$ from the set $L^{(t)}$ of light edges seen so far in the graph stream.   

\section{\algname: Counting Triangles with Predictions}
\label{sec:alg}

We now describe our algorithm \algname\ for counting triangles in graph streams with predictions. Again, due to space constraints and for the sake of clarity, we describe the version for insertion-only streams (our algorithm for fully dynamic streams is in Appendix~\ref{sec:additionalalgos}). Similarly to previous approaches, \algname\ uses reservoir sampling and waiting room sampling, but in addition it devotes part of its memory to store edges predicted as \emph{heavy} by a predictor. In particular, \algname\ stores  a set $\mathcal{S}$ of at most $k$ edges in memory, where $k$ is provided in input. The set $\mathcal{S}$ comprises three disjoint sets of edges: the waiting room $\mathcal{W}$, storing the $k  \alpha$ most recent edges in the stream, where $\alpha \in (0,1)$ is constant; the set $\mathcal{H}$ with the heaviest  edges (according to the predictor) observed in the stream, of size $\left|\mathcal{H}\right| = k  \left(1 - \alpha\right)  \beta$, where $\beta \in (0,1)$ is constant; the set $\mathcal{S}_L$, storing a sample of dimension $k  \left(1 - \alpha\right)  \left(1 - \beta\right)$ of \emph{light} edges, i.e., edges observed in $\Sigma^{(t)}$ that are neither in $\mathcal{W}$ nor in $\mathcal{H}$ at time $t$. \algname\ is a 1-pass streaming algorithm, that is, it returns estimates of global and local counts of all the triangles in the graph $G$ at the end of only one pass of edges in the input stream. 

As stated in Section~\ref{sec:pred}, \algname\ employs a predictor $O_H$ that predicts a measure of heaviness for  every edge. In general, \algname\ makes no assumption on the quality of the predictor $O_H$. For example, we do not assume that $O_H$ provides reliable measure of the number $\Delta(e)$ of triangles that include edge $e$ for all $e$. Our algorithm provides unbiased estimates independently of the quality of the predictor $O_H$, and our analysis (see Section~\ref{sec:analysis}) shows that, as long as the total number of triangles captured by edges kept in $\mathcal{H}$ (due to the predictions of $O_H$) is large enough (compared to the total number of triangles), the use of $O_H$ leads to improved estimates of the number of triangles. Moreover, we prove that if the predictor $O_H$ makes random predictions, our algorithm does not provide worse estimates than previous approaches.

We now describe \algname; its pseudocode is presented in Algorithm~\ref{alg:tonic-ins}. \algname\ first initializes  the sets $\mathcal{W}, \mathcal{H}$, and $\mathcal{S}_L$, the estimate $\hat{T}$ of the global number of triangles, and the observed number $\ell$ of light edges (lines~\ref{line:1}-\ref{line:2}). Let $t$ be the instant of time in which the edge $\{u, v\}$ arrives in the stream (line~\ref{line:3}). Let $\mathcal{S}^{(t)}$ denote the set of edges present in $\mathcal{S}$ at time $t$, and analogously $\mathcal{W}^{(t)}$, $\mathcal{H}^{(t)}$ for the waiting room $\mathcal{W}$ and the heavy edges $\mathcal{H}$, respectively.  At time $t$, the algorithm performs the following steps. First, it counts each occurrence of triangles containing $\{u, v\}$ in the set $\mathcal{S}^{(t)}$ (line~\ref{line:alg2counttriangles}), computing for each triangle $\Delta~=~\{u,v,w\}$ a corresponding probability $p_\Delta$ according to whether $\{u,w\}$ and/or $\{v,w\}$ belong to the sample $\mathcal{S}_L$ of light edges or not; $p_\Delta$ is used as a correction factor for updating  the estimated counts of the global and local number of triangles. \algname\ then updates $\mathcal{W}$, $\mathcal{H}$, and $\mathcal{S}_L$ as follows. If $\mathcal{W}$ is not full, it inserts $\{u, v\}$ in $\mathcal{W}$ (line~\ref{line:alg2WR1}). If $\mathcal{W}$ is full, it replaces the oldest edge $\{x, y\}$ in $\mathcal{W}$ with $\{u, v\}$ (lines~\ref{line:7}-\ref{line:alg2WR2}), and, if $\mathcal{H}$ is not full, it inserts $\{x, y\}$ in $\mathcal{H}$ (lines~\ref{line:9}-\ref{line:alg2Hnotfull}). Otherwise, we are going to observe a light edge and increment counter $\ell$ (line~\ref{line:alg2ell}). Let $e'$ be the \emph{lightest} edge in $\mathcal{H}$, i.e., $e' = \arg\min_{e\in \mathcal{H}} O_H(e)$ (line~\ref{line:alg2lightestHE}). Let $e_{\max}$ be the heaviest edge between $\{x,y\}$ and $e'$, and let $e_{\min}$ be the lightest edge between such edges. Then \algname\ keeps $e_{\max}$ in $\mathcal{H}$ and updates $\mathcal{S}_L$: if $\mathcal{S}_L$ is not full, it inserts $e_{\min}$ in $\mathcal{S}_L$; otherwise, it updates $\mathcal{S}_L$ using reservoir sampling for edge $e_{\min}$ (lines~\ref{line:14}-\ref{alg:sampleedge}).  \algname\ reports the final estimates $\hat{T}$ and $\hat{T}_u$ at the end of the stream (line~\ref{alg:return}).

\begin{algorithm}[htbp]
	\footnotesize
	\caption{\algname\  $\left(\Sigma, k, \alpha, \beta, O_H \right)$}
	\label{alg:tonic-ins}
	\LinesNumbered
	\kwInput{Arbitrary order edge stream $\Sigma = \{e^{(1)}, e^{(2)}, ... \}$; memory budget $k$; fraction of waiting room space $\alpha$; fraction of heavy edges space $\beta$; edge heaviness predictor $O_H$.}
	\kwOutput{Estimate of global triangles count $\hat{T}$; 
		estimate of local triangles count $\hat{T_u}$ for each node $u$.}
	$\mathcal{W} \longleftarrow \emptyset$; $\mathcal{H} \longleftarrow \emptyset $; $\mathcal{S}_L \longleftarrow \emptyset $\label{line:1}\;
	$ \hat{T} \longleftarrow 0$; $ \ell \longleftarrow 0$\label{line:2}\;
	\For{each edge $e^{(t)} = \{u, v\}$ in the stream $\Sigma$\label{line:3}}{
		\texttt{CountTriangles}$\left( \{u, v\}, \mathcal{W} \cup \mathcal{H} \cup \mathcal{S}_L,  \ell \right)$\label{line:alg2counttriangles}\;
		\lIf{$\left|\mathcal{W}\right| < k\alpha$}{$\mathcal{W} \longleftarrow \mathcal{W} \cup \{\{u, v\}\}$} \label{line:alg2WR1}
		
		\Else{
			$\{x, y\} \gets $ oldest edge in $\mathcal{W}$\label{line:7}\;
			$\mathcal{W} \longleftarrow \mathcal{W} \setminus \{\{x, y\}\} \cup \{\{u, v\}\}$\label{line:alg2WR2}\;
			\If{$\left|\mathcal{H}\right| < k \left(1 - \alpha\right) \beta$\label{line:9}}{$\mathcal{H} \longleftarrow \mathcal{H} \cup \{\{x, y\}\}$;}\label{line:alg2Hnotfull}
			
			\Else{
				$ \ell \longleftarrow \ell + 1$\label{line:alg2ell}\;
				$ \{u', v'\} \longleftarrow$ lightest edge in $\mathcal{H}$\label{line:alg2lightestHE}\;  
				\uIf{$O_H\left(\{x, y\}\right) >  O_H\left(\{u', v'\}\right)$\label{line:14}}{
					$\mathcal{H} \longleftarrow \mathcal{H} \cup \{\{x, y\}\} \setminus \{ \{u', v'\}\}$\label{line:alg2H}\;
				}
				\lElse{$\{u', v'\} \longleftarrow \{x, y\}$}
				\uIf{$\ell < k\left(1 - \alpha\right)\left(1 - \beta\right)$}{ \label{line:alg2samplenotfull1}
					$\mathcal{S}_L \longleftarrow \mathcal{S}_L \cup \{\{u', v'\}\}$\label{line:alg2samplenotfull2}\;
				}
				\lElse{\texttt{SampleLightEdge}$\left( \{u', v'\}, \ \mathcal{S}_L,  \ell \right)$\label{alg:sampleedge}}
			}
		}
	}
	\Return $\hat{T}$, $\hat{T}_u$ for each node $u \in \mathcal{S} = \mathcal{W} \cup \mathcal{H} \cup \mathcal{S}_L$\label{alg:return}\;
	
\end{algorithm}

At any time $t$, $\hat{T}$ provides an estimate of the (global) number of triangles observed in the graph $E^{(t)}$ up to time $t$. For local triangles, \algname\ maintains estimates $\hat{T}_u >0$ for some vertices $u$ in $V$, in particular for vertices with triangles contributing to $\hat{T}$; for all other vertices, the estimated number of triangles is $0$.

\algname\ makes use of two subroutines, \texttt{CountTriangles} and \texttt{SampleLightEdge}.  \texttt{CountTriangles} (Alg.~\ref{alg:counttriangles2})  counts the number of triangles \emph{closed} by the current edge in the stream, computes the probability that such triangles have been sampled by the algorithm and uses such probability to update the relevant counts (see Sect.~\ref{sec:prob_comp} in Appendix for the correctness of the probabilities computed). \texttt{SampleLightEdge} (Alg.~\ref{alg:samplelightedge}) uses reservoir sampling (see Section~\ref{sec:reservoirs}) to update the sample $\mathcal{S}_L$ of the set $L$ of light edges observed in stream up to that time.
Also, we assume that \textit{FlipBiasedCoin(p)} flips \textit{HEAD} with probability $p$.

\subsection{Analysis}
\label{sec:analysis}

In this section, we analyze our algorithm \algname. In particular, we first prove that \algname\ provides unbiased estimates of the number of global/local triangles at each time step. We then analyze the time  complexity of our algorithm, and also analytically assess when the use of a noisy predictor leads to estimates with smaller variance (i.e, that are more accurate) compared to the \wrs\ algorithm from~\cite{shin2017wrs}. For lack of space, all proofs and analyses, including the ones for our algorithm for fully dynamic streams, are in Appendix~\ref{sec:proofs}.

\begin{algorithm}[htbp]
	\footnotesize
	\caption{\texttt{CountTriangles}$\left(\{u, v\}, \mathcal{S}, \ell \right)$}
	\label{alg:counttriangles2}
	\LinesNumbered
	\kwInput{edge $\{u, v\}$; subgraph $\mathcal{S} = \left(\hat{V}, \ \hat{E}\right)$; number of predicted light edges $\ell$.
	}
	\For{each node $w$ in $\hat{\mathcal{N}}_u \cap \hat{\mathcal{N}}_v$}{ \label{line:alg2commonneighs}
		initialize $\hat{T_u}$, $\hat{T_v}$, $\hat{T_w}$ to zero if not set yet\;
		$p_{uvw} = 1$\;
		\uIf{$\{w, u\} \in \mathcal{S}_L$ AND $\{v, w\} \in \mathcal{S}_L$}{
			$p_{uvw} = \min \left( 1, \ \frac{k(1 - \alpha)(1 - \beta)}{\ell} \times \frac{k(1 - \alpha)(1 - \beta) - 1}{\ell - 1} \right)$\;
		} \uElseIf{$\{w, u\} \in \mathcal{S}_L$ OR $\{v, w\} \in \mathcal{S}_L$}{
			$p_{uvw} = \min \left( 1, \frac{k(1 - \alpha)(1 - \beta)}{\ell} \right) $\;
		}
		increment $\hat{T}$, $\hat{T}_u$, $\hat{T}_v$, $\hat{T}_w$ by $ 1 / p_{uvw}$\;        
	}
	
\end{algorithm}

\begin{algorithm}[htbp]
	\footnotesize
	\caption{\texttt{SampleLightEdge}$\left(\{u, v\}, \mathcal{S}_L, \ell \right)$}
	\label{alg:samplelightedge}
	\LinesNumbered
	\kwInput{edge $\{u, v\}$; current sample $\mathcal{S}_L$ of light edges; number of observed light edges $\ell$ in the stream.
	}
	$p_{sampling} = \frac{k\left(1 - \alpha\right)\left(1 - \beta\right)}{\ell}$\;\label{line:alg2probsamp}
	\If{FlipBiasedCoin$\left(p_{sampling}\right)$ == HEAD}{
		$\{\bar{u}, \bar{v}\} \longleftarrow$ edge sampled uniformly at random from $\mathcal{S}_L$\;
		$\mathcal{S}_L \longleftarrow \mathcal{S}_L  \setminus \{\{\bar{u}, \bar{v}\}\} \cup \{\{u, v\}\}$\label{line:samplelightSL}\;} 
\end{algorithm}

\subsubsection{Unbiasedness}
\label{sec:unbiasedness}
The unbiasedness of the estimates reported by \algname\ is provided by the following result.

\begin{theorem}
	\label{thm:unbiasedness}
	Let $T^{(t)}$ and $T_u^{(t)}$ be the true global count of triangles in the graph and the true local triangle count for node $u \in V$ at time $t$, respectively. We have:
	\begin{equation*}
		\mathbb{E}\left[\hat{T}^{(t)}\right] = T^{(t)}, \ \forall \ t \geq 0
	\end{equation*}
	\begin{equation*}
		\mathbb{E}\left[\hat{T}_u^{(t)}\right] = T_u^{(t)} \ \forall u \in V, \ \forall \ t \geq 0
		\label{eq:localcount2}
	\end{equation*}
\end{theorem}

\subsubsection{Time Complexity}

In the following, we analyze time complexity of \algname. In the analysis we also account for how the sets of edges stored by \algname\ are implemented. The worst-case time complexity of \algname\ is dominated by computing the number of triangles the last observed edge closes, plus a logarithmic factor due to retrieving, for each edge in the graph stream, the lightest edge from the set of $\mathcal{H}$. 

\begin{theorem}
	\label{theorem:timecomplexity2}
	Given an input graph stream $\Sigma$ of insertion-only edges, and given $\alpha = \left|\mathcal{W}\right|/k$, $\beta = \left|\mathcal{H}\right|/k(1 - \alpha)$, \algname\ processes each edge in $\Sigma$ in $\bigO{\left( k + \log \left(k (1 - \alpha) \beta \right) \right)}$ time.
\end{theorem}

From the above, the total time to compute the estimation of local and global triangles for the entire graph stream of $m$ edges is $\bigO{\left( mk + m\log \left(k (1 - \alpha) \beta \right) \right)}$. This is a worst-case bound since in practice, in real graph streams, the complexity of the computation of common neighbors for nodes $u$ and $v$ is much smaller than $k$, i.e., our memory budget.

\subsubsection{Comparison with \wrs\ for global triangles count}
\label{sec:analysis_variance_compariso}
We now prove that our algorithm leads to better estimates than \wrs~\cite{shin2017wrs} for the global number of triangles when the predictions from the predictor $O_H$ are \emph{useful}, in the sense that they lead to consider as \emph{heavy}, edges that are involved in a large number of triangles. Note that this is different from having a \emph{perfect oracle} (i.e., making no errors).

For the sake of simplicity, in the analysis we consider a simplified version of the \wrs\ algorithm and a simplified version of \algname\ algorithm that capture the main features of the two approaches. The simplified version of \wrs\ samples each light edge (i.e, that leaves the waiting room) independently with probability $p$;  the simplified version of \algname\  uses a predictor that predicts an edge leaving the waiting room as \emph{heavy} or \emph{light}, keeps (predicted) heavy edges in $\mathcal{H}$, and samples each light edge (see line~\ref{alg:sampleedge} of Algorithm~\ref{alg:tonic-ins}) with probability $p' <p$, where $p$ is the probability that an edge is sampled by \wrs. Note that by properly fixing $p'$, accounting for the memory budget for heavy edges (i.e., $k \left(1 - \alpha\right)\beta$), the two algorithms have the same expected memory budget. 
We assume that the waiting room has the same size in both \wrs\ and \algname. Note that triangles where the first two edges are in the waiting room $\mathcal{W}$ at the time of discovery are counted (with probability 1) by both algorithms. Since the size of $\mathcal{W}$ is the same, the sets of triangles counted (with probability 1) in the waiting room by the two algorithms are identical. Given that we are interested in the difference in the estimates from the two algorithms, we exclude such triangles from our analysis (i.e., all quantities below are meant after discarding such triangles).

Note that the quality of the approximation reported by our algorithm \algname\ w.r.t. \wrs\ must depend on several quantities: the number of triangles in which heavy edges are involved; the number of triangles in which light edges are involved; $p$ and $p'$, that govern the memory allocated for light edges; the quality of the predictor in differentiating heavy edges from light edges. Our result provides a formal relation between such quantities.

Let $\Delta(e)$ be the number of triangles in which edge $e$ is involved.  Let $T^H = \sum_{h \in \mathcal{H}} \Delta(h)$  be the sum, over heavy edges, of the number of triangles in which heavy edges appear; analogously, let $T^L = \sum_{l \in L} \Delta(l)$ be the sum, over light edges, of the number of triangles in which light edges appear.

In our analysis we make some assumptions on the predictor; in particular, we will assume heavy edges appear in $\ge \rho$ triangles, while light edges appear in $< \rho$ triangles, for some $\rho \ge 3$. However, to capture the errors of the predictor, in particular around the threshold $\rho$, we assume that edges $h$ predicted as heavy by the predictor are involved in $\Delta(h)\ge \rho/c$ triangles, while edges $l$ predicted as light by the predictor are involved in $\Delta(l) < c\rho$ triangles, for some constant $c \ge 1$.  Note that for edges $e$ with $\rho / c < \Delta(e) < c\rho$ the predictor can make arbitrarily wrong predictions.

\begin{proposition}
	\label{prop:var_vs_WRS}
	Let $\Var[\hat{T}_{\wrs}(p)]$ be the variance of the estimate  $\hat{T}_{\wrs}$ obtained by \wrs\ when light edges are sampled independently with probability $p$, and let  $\Var[\hat{T}_{\algname}(p')]$ be the variance of the estimate  $\hat{T}_{\algname}$ obtained by \algname\ when light edges are sampled with probability $p'$. Then $\Var[\hat{T}_{\algname}(p')] \le \Var[\hat{T}_{\wrs}(p)]$ if 
	\begin{equation*}
		\frac{T^H}{T^L} > 3 \frac{(1/p'^2 - 1/p^2)+c\rho(1/p' - 1/p)}{(1/p-1)(3+4\rho/c)}.
	\end{equation*}
\end{proposition}

The result above explicits a trade-off between the quality of the predictor (represented by $c$), the impact of heavy edges in the count (represented by $\rho$ and $T^H$) with respect to light edges (represented by $T^L$), and the fraction of memory allocated for heavy edges in our algorithm (that depends on the difference between $p$ and $p'$.) For example, if $p=0.1$ and $p'=0.09$ (these are representative values for our experimental evaluation), $c=1.5$, and $\rho=10$, then the bound above is $\frac{T^H}{T^L} > 0.45$, that corresponds to the sum of the counts for heavy edges being at least one third of all triangles (that do not have two edges in the waiting room $\mathcal{W}$ when counted); if instead the parameters are as above but $\rho=100$, then $\frac{T^H}{T^L} > 0.24$, that corresponds to the sum of the counts for heavy edges being at least one fifth of all triangles (again, that do not have two edges in the waiting room $\mathcal{W}$ when counted). For the latter case, note that $c=1.5$ implies that the predictor is not very accurate: a light edge $l$ (with $\Delta(l)<100$ by definition) may be mispredicted as heavy as long as $\Delta(l)>66$, while a heavy edge $h$ (with $\Delta(h) \ge 100$ by definition) may be mispredicted as light as long as $\Delta(h)<150$.

The result above formalizes the intuition that the predictor helps when it provides fairly reliable information on heavy edges. The following result instead proves that, when the predictor does not provide useful information on heavy edges, our algorithm returns estimates as accurate as \wrs, in the sense that the variances of the estimates from the two algorithms are the same when the same amount of memory is used.

\begin{proposition}
	\label{prop:variance_random_pred}
	If the predictor $O_H$ predicts a random set of edges as heavy edges, and \wrs\ and \algname\ use the same amount of memory, then $\Var[\hat{T}_{\algname}(p')] = \Var[\hat{T}_{\wrs}(p)]$.
\end{proposition}

Strikingly, our experimental evaluation shows that our algorithm \algname\ provides estimates of quality similar to \wrs\ even when the most \emph{adversarial} predictor is provided (see Sect.~\ref{sec:adversarial-experiments} in Appendix).

\subsection{A Simple Domain-Independent Predictor}
\label{sec:simple_pred}
The predictor $O_H$ used by \algname\ could be implemented by using one of the several machine learning models that may consider information other than the graph $G$ in its prediction. For example, in social networks, information about the users may be useful to predict the heaviness of an edge. In protein interaction networks, the function and properties of the protein provide a lot of information on its heaviness. However, we now describe a simple predictor based only on the graph structure that, as we will show, is extremely powerful in practice.

\sloppy{
	We define a simple predictor  \texttt{MinDegreePredictor} that predicts, as measure of heaviness for $\{u,v\}$, the minimum between the degree $deg(u)$ of $u$ and the degree $deg(v)$ of $v$, that is $O_H~=~\min\left\{ deg(u), deg(v) \right\}$. Note that the degree $deg(u)$ for each node $u\in V$ can be computed easily with 1 pass on the data whenever the whole stream is available beforehand. Moreover, in practical applications, one can measure the number of observed edges involving a given node $u$ in a \emph{training} phase, and then use such information for the prediction in later phases. Additionally, to reduce the memory required for the predictor, instead of storing \emph{all} nodes degrees, one can simply store the largest ones, and assume a value of $0$ for the degree of all other nodes.
}

\section{Experimental Evaluation}
\label{sec:experiments}
In this section we present the results of our experiments. Due to space constraints, we only present a subset of results for the estimation of the global number of triangles for insertion-only streams. The complete results, including the estimation of the local number of triangles and the results for fully dynamic streams, are in Appendix~\ref{sec:additional-experiments}.

The goals of our experiments are: i) to assess the dependency of  \algname\ from the relative dimension of the  waiting room (parameter $\alpha$), from the relative dimension of the heavy edges set ($(1-\alpha)\beta$), and from the total memory budget ($k$);  ii) to assess the improvement in the estimates that results from the predictor; iii) to compare \algname\ with state-of-the-art algorithms (with and without predictors) for counting triangles in edge streams on single streams and on sequences of edge streams, where the predictor is learned only in the first stream (i.e., graph)  of the sequence; iv) to assess the quality of \algname's estimates during the evolution of the input stream.

\textit{Experimental Setup.} We implemented  \algname\ in C\texttt{++}17. The code to reproduce all experiments is available at \url{https://github.com/VandinLab/Tonic}. All the code was compiled under \texttt{gcc} 9.4.0 and ran on a 
2.20 GHz Intel Xeon CPU with 1 TB of RAM, on Ubuntu 20.04. If not explicitly specified, we fixed the memory budget of each algorithm to $k = m / 10$ and for each run we report mean and standard deviation across 50 independent trials. To measure accuracy for global triangles estimates, we considered the \emph{global relative error} at the end of the stream, that is $|\hat{T} - T|/T$.

\textit{Datasets.}
We considered both single graph and sequence of graph streams as datasets of interest, representing social and citation networks, and autonomous system (AS) relationships, downloaded from \cite{nr, konect, snapnets}.  Datasets' names and statistics are summarized in Table \ref{tab:datasets}.  A complete description of the considered dataset can be found in Appendix~\ref{sec:dataset_description}. From each dataset, we removed self-loops and multiple edges, deriving a stream of insertion-only edges, for consistency with previous works. 
For graph sequences (i.e., the last four rows of the table), the statistics in Table~\ref{tab:datasets} refer to the graph with highest number of nodes. 

\begin{table}[htbp]
	\centering
	\footnotesize
	\begin{tabular}{ c c c c c }
			Dataset & $n$ & $m$ & $T$ \\
			\toprule
			\multicolumn{5}{c}{\emph{Single Graphs}} \\
			\midrule
			Edit EN Wikibooks    										 & $133k$ 					   & $386k$  		  				  & $178k$ 				  \\
			SOC YouTube Growth    									& $3.2M$ 					 & $9.3M$  		  				    & $12.3M$ 			   \\
			Cit US Patents    										 		& $3.7M$ 					  & $16.5M$  		  				& $7.5M$ 			    \\
			Actors Collaborations   								   & $382k$ 					 & $15M$  		  				    & $346.8M$ 			 \\
			Stackoverflow    										 	    & $2.5M$ 					 & $28.1M$  		  				& $114.2M $ 		  \\
			SOC LiveJournal    										      & $4.8M$ 					   & $42.8M$  		  				 & $285.7M$ 		  \\
			Twitter-merged      										   & $41M$ 					   & $1.2B$  		  				    & $34.8B$ 		        \\
			\midrule
			\multicolumn{4}{c}{\emph{Snapshot Sequences}} \\
			\midrule
			Oregon (9 graphs)   				   					   & $11k$ 	  					  & $23k$  		 				       & $19.8k$     		   & 			  		\\
			AS-CAIDA (122 graphs)   	 	  					  & $26k$ 	 					& $53k$  	   						 & $36.3k$ 		  	    &				  \\
			AS-733 (733 graphs)   									& $6k$ 	    				   &$13k$  	 	   					     & $6.5k$ 		  		 & 		 			    \\
			Twitter (4 graphs)   									    & $29.9M$ 	    		    &$373M$  	 	   				   & $4.4B$ 		  	  & 		 			    \\
			\bottomrule
		\end{tabular}
	\caption{Datasets' statistics: number $n$ of nodes; number $m$ of edges; number $T$ of triangles}
	\label{tab:datasets}
\end{table}

\textit{Predictors.}
For our algorithm \algname, the predictors used depend on whether we analyze a single graph stream or a sequence of graph streams. For single streams, we considered three predictors: \texttt{OracleExact}, where $O_H(e)$ is the number $\Delta(e)$ of triangles involving $e$ in the graph stream (i.e., the heaviness); \texttt{Oracle-noWR}, where $O_H(e)$ is obtained subtracting to $\Delta(e)$  the number of triangles for which $e$ is in the waiting room $\mathcal{W}$; \texttt{MinDegreePredictor}, the predictor described in Section~\ref{sec:simple_pred} for which $O_H(\{u,v\})$ is the minimum between the degree of $u$ and the degree of $v$. Note that, when analyzing single streams,  \texttt{OracleExact} and \texttt{Oracle-noWR} are mostly \emph{ideal} and not \emph{practical} predictors (since they require to solve the counting problem exactly first), but we use them to study the potential gain obtained with a predictor. In addition, \texttt{OracleExact} represents a general predictor for heaviness, while \texttt{Oracle-noWR} is a predictor tied to the counting strategy employed by our algorithm. In contrast, \texttt{MinDegreePredictor} is a predictor that can be easily implemented with a single pass on the edge stream, and is therefore practical when the entire edge stream is available beforehand.
For  \texttt{OracleExact} and \texttt{Oracle-noWR}, the predictor maintains the top 10\% edges sorted by predicted values (i.e., the top $m/10$ entries $((u,v);O_H(\edge{u}{v}))$ sorted by decreasing $O_H(\cdot)$). Hence, we consider predictors (practical or ideal) that provide accurate information only for the top 10\% edges, while for the remaining 90\%  of edges they output $O_H(e)=0$.
For \texttt{MinDegreePredictor}, we use a different representation based on nodes. In particular, the predictor stores $\bar{n}$ entries $(u ; deg(u))$ corresponding to the $\bar{n}$ highest degree nodes in the graph. 
For an edge $e = \edge{u}{v}$, \texttt{MinDegreePredictor} produces in output $\min{(deg(u), deg(v))}$ if both $(u; deg(u))$ and $(v;deg(v))$ are stored in the predictor, and $0$ otherwise. The value $\bar{n}$ is fixed so to correspond to the number of (unique) nodes that would be required to compute the exact value for the top $10\%$ edges (according to the  \texttt{MinDegreePredictor}); note that the nodes stored by the predictor may actually be different from the ones required to predict the exact (min-degree) value for the top $10\%$ edges, since we are storing the highest degree nodes. 
As a result, \texttt{MinDegreePredictor}  stores significantly less entries than  \texttt{OracleExact} and  \texttt{Oracle-noWR} (see results below). Further details can be found in Appendix~\ref{sec:predictors_details}.


For sequences $\Sigma_1, \Sigma_2, \dots$ of graph streams, we considered the \emph{same} predictors, but \emph{trained} using only the \emph{first} stream  $\Sigma_1$. For example, in \texttt{OracleExact}, $O_H(e)$ is the number of triangles of $\Sigma_1$ in which  $e$ appears; in later streams, edges $e$ adjacent to vertices not in $\Sigma_1$  always have $O_H(e)=0$. 
Note that \texttt{OracleExact}/\texttt{Oracle-noWR} can be obtained by solving the problem exactly on the graph stream $\Sigma_1$, and are therefore \emph{practical} whenever this is feasible. 

\textit{Baselines.}
We compared our algorithm \algname\ with \wrs~\cite{shin2017wrs, lee2020temporal} (considering the algorithm for insertion-only graph streams), that is the state-of-the-art for global and local triangles in (insertion-only) graph streams, and with the algorithm from~\cite{chen2022triangle} (considering the algorithm for arbitrary order streams), that we denote as \chenalg, which is the only algorithm that uses predictors and is limited to estimate global triangles at the end of the (insertion-only) stream\footnote{While the algorithm described in \cite{chen2022triangle} uses fixed probability sampling, the implementation provided by the authors  enforces a maximum memory budget by essentially removing a random light edge from main memory, that does not provide guarantees on the resulting sample.}. For fully dynamic (FD) streams, we compared our algorithm \algnamedel\ with \wrsdel~\cite{lee2020temporal} (considering the version for FD streams) and with \thinkdacc\ (the algorithm with fixed memory from ~\cite{shin2018think, shin2020fast}), which are both state-of-the-art for global and local triangles in (FD) graph streams. In all experiments, all algorithms are provided with the same memory budget. The main parameter of \wrs\ and \wrsdel\ is the fraction $\alpha$ of the memory allocated to the waiting room, while the main parameter for \chenalg\ is the fraction $\beta$ of the memory  allocated to the heavy edges. For \chenalg\,  we used as predictor \texttt{OracleExact}, that is, for single streams we are considering the best predictor that it could have access to. \thinkdacc\ has no tunable parameters. 

\textit{Results.} 
First, we ran \algname\ for various values of the parameters $\alpha$ and $\beta$ (see Sect.~\ref{sec:accuracy-vs-params} in Appendix for more details). We also ran \wrs\ and \chenalg\ with various values of $\alpha$ (for \wrs) and $\beta$ (for \chenalg).

\begin{figure*}[htbp]
	\centering
	\includegraphics[width=1 \textwidth]{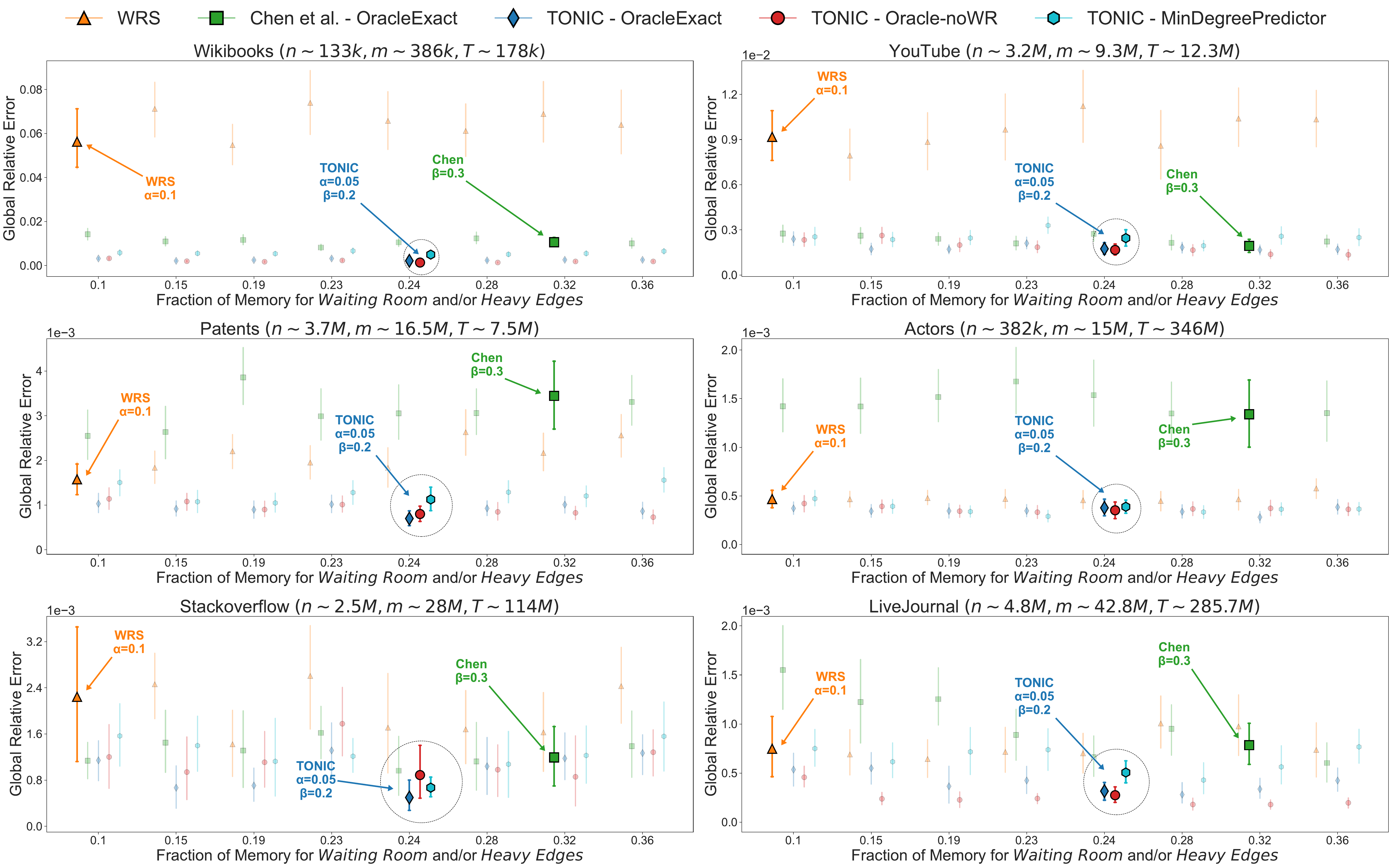}
	\caption{Error vs fraction of memory budget used for waiting room and/or heavy edges. All algorithms are provided with memory budget $k = m/10$. For each combination of algorithm and parameter (including predictor for \algname), the average and 95\% confidence interval over 50 repetitions are shown. The chosen configuration for \algname\ and the configurations suggested by \wrs\ and \chenalg\ publications are highlighted.}
	\label{fig:accuracy_params_experiments_merged}
\end{figure*}

 Fig.~\ref{fig:accuracy_params_experiments_merged} shows the error as a function of the fraction of memory budget used for the waiting room and/or heavy edges ($\alpha +(1 - \alpha)\beta$ for \algname, $\alpha$ for \wrs, and $\beta$ for \chenalg) for all our single graphs but the Twitter ones, due to time impracticability on such large graphs. We observe that  \wrs\ and \chenalg\ have performance that strongly depends on the dataset and the parameters, while our algorithm \algname\ provides estimates with low error for all values of $\alpha$ and $\beta$, with the best combination depending on the dataset, and for all datasets. However, overall the combination $\alpha = 0.05$ and $\beta = 0.2$ leads to good results across all datasets, and in all subsequent experiments we fixed the parameters to such values.
Note that \algname\ with $\alpha = 0.05$ and $\beta = 0.2$  is always better than \wrs\ and \chenalg\ with parameters as suggested in the respective publications ($\alpha = 0.1$ and $\beta = 0.3$). Interestingly, \algname\ with the \texttt{MinDegreePredictor} (the only practical one) is always close to \algname\ with predictors \texttt{OracleExact} and \texttt{Oracle-noWR}, and always outperforms \wrs\ and \chenalg\ other than for Youtube, where \chenalg\ (empowered by the \texttt{OracleExact}) shows slightly lower error. These results show that our algorithm \algname\ is robust to the choice of the parameters $\alpha$ and $\beta$, and provides accurate estimates consistently across all datasets even when a simple practical predictor is used, adapting to the data distribution thanks to the use of the predictor.

We then assessed how \algname\ behaves in terms of  error and runtime as a function of the memory budget $k$, and compared it with \wrs\ and \chenalg. For every dataset, we ran all algorithms with memory budget $k = f  m$, with $m=|E|$ and for values of $f$ showed on the x-axis. Fig.~\ref{fig:memory_experiments_merged} (left)  shows the estimation errors, with \algname\ achieving a better accuracy in almost every case.  Fig.~\ref{fig:memory_experiments_merged} (right) shows the runtime only for \algname\ and \wrs; the runtime of \chenalg\ is always at least 4, and up to 15, times larger than \algname\ runtime and therefore not shown in Fig.~\ref{fig:memory_experiments_merged} (right) for the sake of clarity. 

\begin{figure*}[htbp]
	\centering
	\includegraphics[width=1 \textwidth]{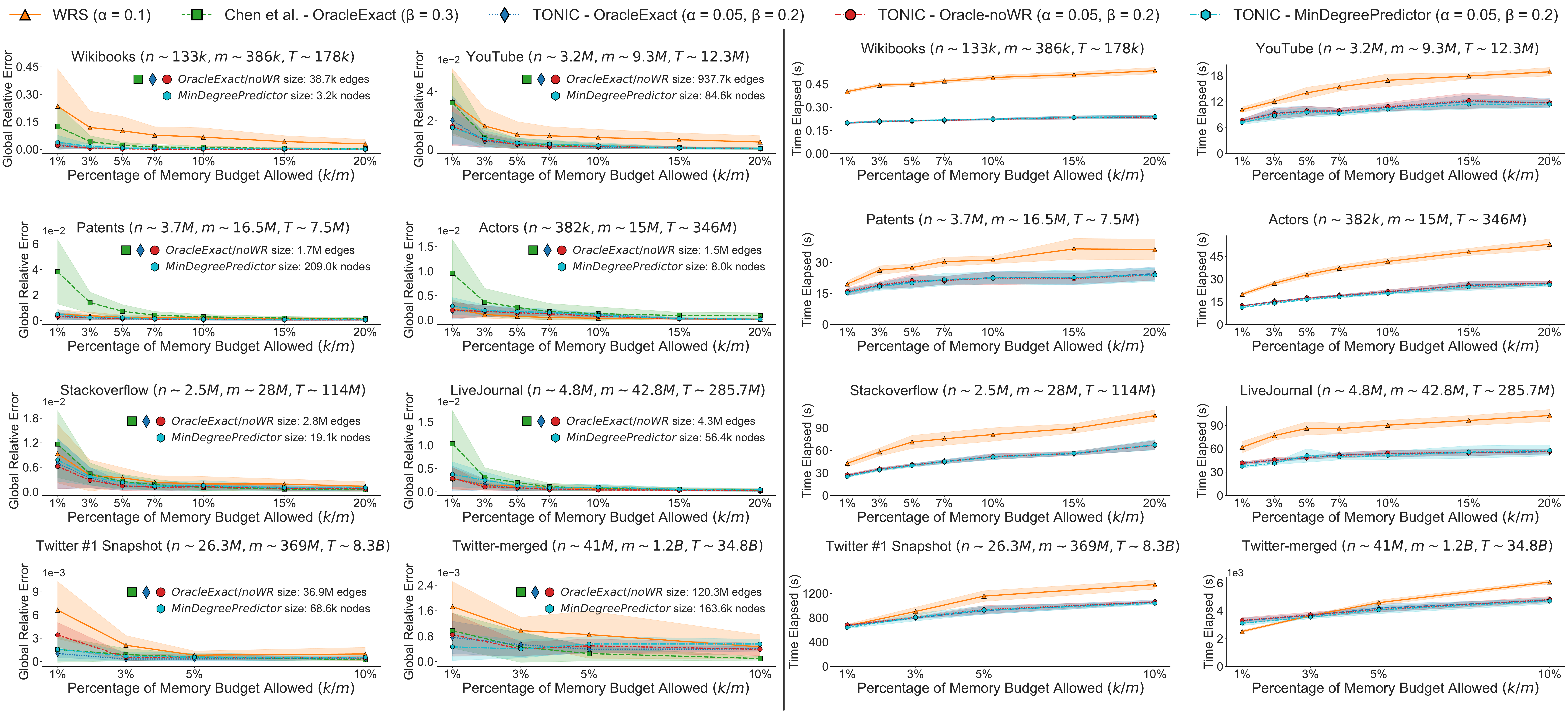}
	\caption{Error (left) and runtime (right) vs  memory budget. Each left subplot reports the size of \texttt{OracleExact} and \texttt{Oracle-noWR} as number of edges, and the size of \texttt{MinDegreePredictor} as number of nodes. \chenalg\ runtimes are not shown for clarity, since they are 4-15 times bigger than \algname\ runtime. For each combination of algorithm and parameter (including predictor for \algname), the average and standard deviation over 50 repetitions are shown. The algorithms parameters are as in legend (for \wrs\ and \chenalg\ they are fixed as in the respective publications; for \algname\ they are as chosen in Fig.~\ref{fig:accuracy_params_experiments_merged}).} 
	\label{fig:memory_experiments_merged}
\end{figure*}

We see that \algname\ is always faster than \wrs\ (excluding $1\%m$ memory budget for Twitter-merged dataset, for which the runtime is slightly higher, but still reasonable, while \algname\ significantly outperforms \wrs\ in terms of approximation error), showing a much milder slope and hence, in practice, able to scale better with respect to worst-case analyses of the running time. For larger memory budgets, \algname\ usually outperforms \wrs\ in terms of approximation error (or is at least comparable to it), while being faster. Therefore, in all scenarios \algname\ provides an advantage over \wrs. On Twitter-merged for memory budgets $k= 5\% m$ and $k= 10\%m$, \chenalg\ has the best accuracy, slightly better than \algname. This is due to Twitter streams being in adjacency list order, since in such cases the waiting room is not beneficial. However, while the accuracy of \algname\ is worse than, but comparable to, \chenalg,
the difference in runtime is huge: for a single run with $k = 10\% m$, \algname\ takes $\sim 70$ minutes, while \chenalg\ requires more than $16$ hours. 
These results show that overall \algname\  provides more accurate estimates than other methods, while being faster, also when using the practical \texttt{MinDegreePredictor}. We point out that, as shown in Fig.~\ref{fig:memory_experiments_merged} (left), \texttt{MinDegreePredictor} stores significantly less entries than \texttt{OracleExact} and \texttt{Oracle-noWR}, while maintaining comparable, and in some cases higher, accuracy.

Moreover, we assessed the quality of the approximations in terms of unbiasedness, variance, quality of estimates at any time of the stream, and number of counted and estimated triangles. Fig.~\ref{fig:theoretical-illustriation} shows the results for some representative datasets (see Appendix~\ref{sec:additional-experiments} for the all results). We observe that, as expected, all three algorithms return unbiased estimates, but \algname\ has a much lower variance (Fig.~\ref{fig:theoretical-illustriation} left), confirming our theoretical analysis (Section~\ref{sec:analysis_variance_compariso}). We also note that  \algname\ returns more accurate estimates than \wrs\ at any time $t$ in the stream (Fig.~\ref{fig:theoretical-illustriation} center-left).  Finally, from  the number of  triangles with 0/1/2 light edges counted and estimated by each algorithm (Fig.~\ref{fig:theoretical-illustriation} center-right, right), we observe that \algname\ leverages both the waiting room and the predictor, leading to an higher number of discovered triangles with low variance (recall that the fraction of memory budget dedicated to ``heavy'' edges is  0.3 for \chenalg\ and 0.19 for \algname ).

\begin{figure*}[htbp]
	\centering
	\includegraphics[width=1\textwidth]{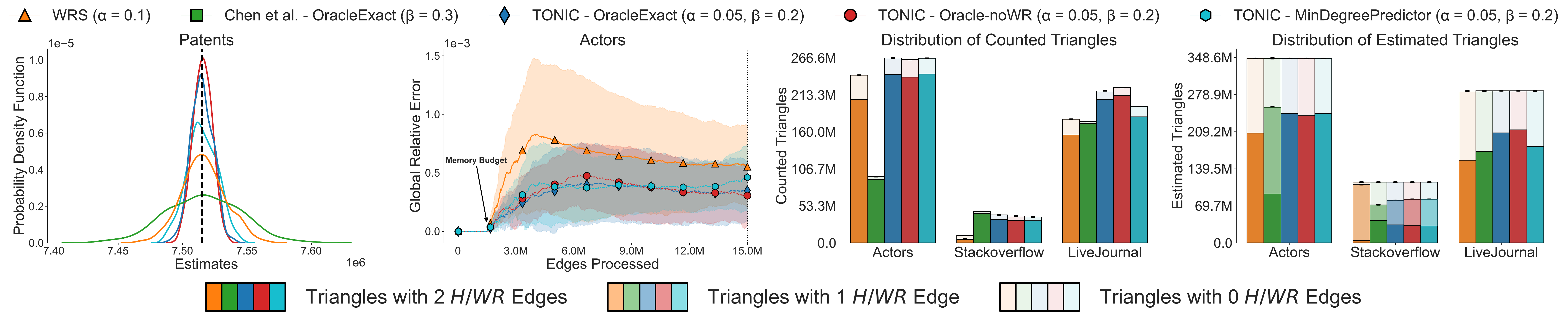}
	\caption{(Left) Distribution of estimates for Patents dataset. (Center-left) Estimation error as time progresses on Actors dataset. (Center-right) Number and type of triangles counted by each algorithm on three datasets. (Right) Fraction of each type of triangles in the total estimates by each algorithm on three datasets.}
	\label{fig:theoretical-illustriation}
\end{figure*}

As final experiments for insertion-only streams, we considered \emph{snapshot networks} with a sequence of graph streams in order to evaluate our algorithm \algname\ in a more challenging and realistic scenario. We considered datasets Oregon, AS-CAIDA and AS-733, including, respectively, 9, 122, and 733 graph streams, and ran each algorithm on each stream of the sequence, with a memory budget equal to 10\% of the number of edges of each considered stream. The predictors used by \algname\ and \chenalg\  are trained \emph{only} with the \emph{first} graph stream, and their predictions are used for each subsequent graph. Note that in this case, for each subsequent graph stream, \texttt{OracleExact} and \texttt{Oracle-noWR} are \emph{imperfect} predictors. Fig.~\ref{fig:snapshot_experiments_merged} reports the error for each graph stream in the sequence. \algname\ achieves outstanding performances with all predictors for all three datasets, with errors that are significantly smaller than \wrs\ and \chenalg\ across all graph streams. For AS-CAIDA and AS-733, with hundreds of graph streams, we observe that the results slightly deteriorate as more streams are considered, due to the fact that later graph streams can be very different from the first one in which the predictor was trained. In particular, for later graphs the data is growing significantly: in some cases, nodes, edges, and number of triangles are almost doubled. These results highlight the usefulness and practical impact of using the learned information from the predictor, as done by our algorithm \algname.

\begin{figure*}[htbp]
	\centering
	\includegraphics[width=1\textwidth]{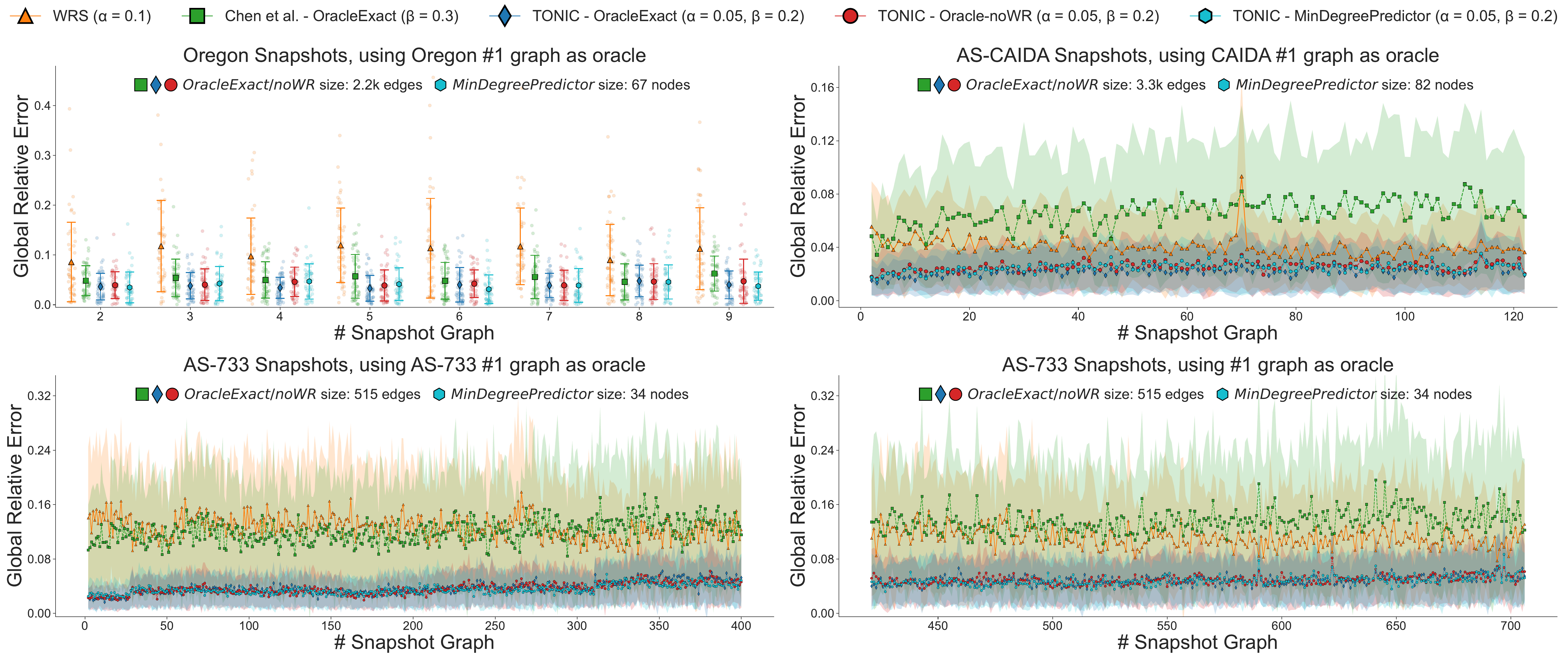}
	\caption{Error with snapshot networks with sequence of graph streams. The bottom plots are for the first 400 streams (left) and the remaining streams (right) of AS-733. In all cases the predictors are trained only on the first graph stream of the sequence (with results not shown on such graph stream). For each combination of algorithm and parameter (including predictor for \algname), the average and standard deviation over 50 repetitions are shown. The  algorithms parameters are as in legend (for \wrs\ and \chenalg\ they are fixed as in the respective publications; for \algname\ they are as chosen in Fig.~\ref{fig:accuracy_params_experiments_merged}). }
	\label{fig:snapshot_experiments_merged}
\end{figure*}

\begin{figure}[htbp]
	\centering
	\includegraphics[width=0.75\textwidth]{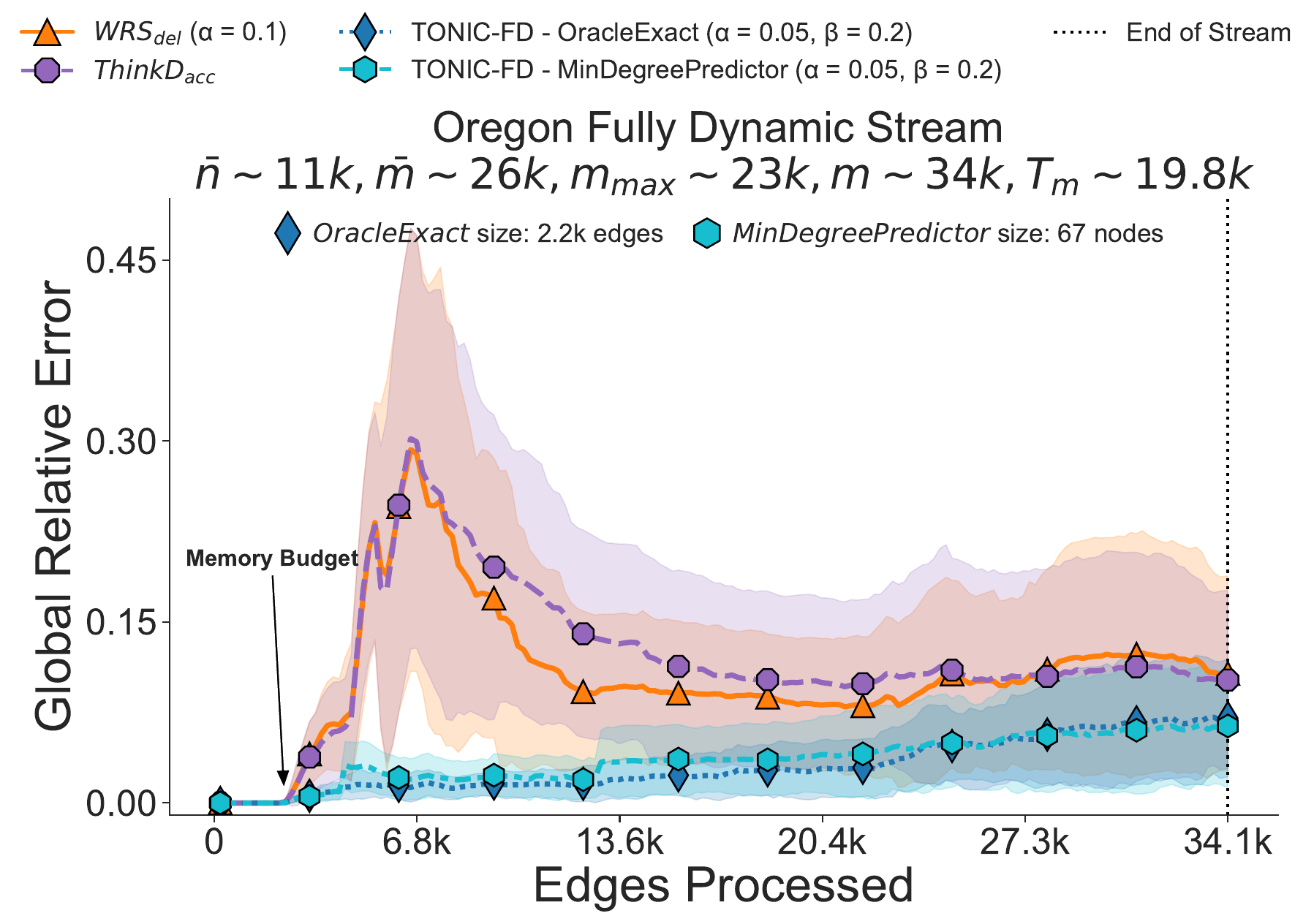}
	\caption{Estimation error as time progresses during Oregon fully dynamic stream. For each combination of algorithm and parameter (including predictor for \algnamedel), the average and standard deviation over 50 repetitions are shown. The algorithms parameters are as in legend (for \wrs\ they are fixed as in the respective publication; for \algnamedel\ they are as chosen in Fig.~\ref{fig:accuracy_params_experiments_merged}). $\bar{n}$: number of unique nodes;  $\bar{m}$: number of unique edges;  $m_{max}$: maximum number of edges at some time; $m$: total number of edges; $T_m$: number of global triangles at the end, derived from the FD stream.}
	\label{fig:fd_experiments}
\end{figure}

Finally, we give a glance of experimental results when dealing with fully dynamic (FD) streams, i.e., with graph streams that allow both insertions and deletions of edges (see Sect.\ref{sec:fd-experiments} in Appendix for all results).  More specifically, we created FD streams starting from our graph sequences datasets: Oregon, AS-CAIDA, AS-733 and Twitter. 
For each dataset, starting from the first graph of the sequence, we computed edge additions and removals between the current and the next snapshot, for all the snapshots in the sequence. 

Edge insertions are added to the FD stream by preserving the temporal ordering following the sequence, and edge deletions are added to the FD stream with random timestamps inside the time window of the snapshots that we are considering. 
We refer to our algorithm for FD streams as \algnamedel\, for which we describe the general workflow in Sect.~\ref{sec:tonc-fd-algo} in Appendix. Predictors for \algnamedel\ are trained \textit{only} on the \textit{first} snapshot of the sequence, in the same setting we described for the above snapshot experiments. Fig.~\ref{fig:fd_experiments} shows the estimation error as time progresses during Oregon FD stream. For all algorithms the memory budget is $k = m_{max}/10$ (highlighted by the black arrow in the stream).
Note that \algnamedel\ always achieves the most accurate estimate at any time through Oregon FD stream. Moreover, as shown in Sect.\ref{sec:fd-experiments} in Appendix, our algorithm \algnamedel\ is always faster than \wrsdel, and faster than or comparable to \thinkdacc, despite requiring removals of edges within our waiting room and heavy edge set.

\section{Conclusion}
In this paper we presented \algname, the first practical algorithm with predictions for fast and accurate triangle counting in insertion-only and fully dynamic graph streams. \algname\ combines waiting room sampling and reservoir sampling with a predictor for the heaviness of edges. We also propose a simple application-independent predictor, based on the degree of the nodes, that can be efficiently computed with 1 pass when the whole stream is available beforehand, and can be easily computed in a training phase in a sequence of graph streams. Our analysis shows that the predictor leads to improved estimates when the edges predicted as heavy do provide useful information, even if the predictor is far from perfect. Our experimental evaluation shows that \algname\ significantly improves the state-of-the-art in terms of error and runtime of the estimates. The improvement is particularly significant in challenging practical scenarios where sequences of hundreds of graph streams are analyzed. Our work opens several directions for future research, including the development of improved predictors.

\section*{Acknowledgment}
\noindent Work supported by “National Centre for HPC, Big Data and Quantum Computing” project, CN00000013 (approved under call M42C – Investimento 1.4 – Avvisto “Centri Nazionali” – D.D. n. 3138 of 16.12.2021, admitted to funding with MUR decree n. 1031 of 06.17.2022), and PRIN Project n. 2022TS4Y3N “EXPAND: scalable algorithms for EXPloratory Analyses of heterogeneous and dynamic Networked Data”

\balance


\appendix

\section{Notation and Symbols}
\label{sec:notations}
Table~\ref{tab:symbols} resumes the notation and symbols used throughout the manuscript, including algorithms and analysis. 

\begin{table}[htbp]
	\centering
	\footnotesize
	\resizebox{1 \textwidth}{!}{
		\begin{tabular}{ l | l }
			\toprule
			\textbf{Symbol} & \textbf{Definition}  \\
			\midrule
			\multicolumn{2}{l}{\textit{Notations for Graph Streams}} \\ 
			\midrule
			$G^{(t)}  = (V, \ E^{(t)})$ & Graph $G$ at time $t$ \\[0.05cm]
			$\{u, v\}$                                                     & (Unordered) edge between $u$ and $v$ \\[0.05cm]
			$\{u, v, +\}$                                                     & Addition of edge $\edge{u}{v}$ \\[0.05cm]
			$\{u, v, -\}$                                                     & Deletion of edge  $\edge{u}{v}$ \\[0.05cm]
			$\{u, v, w\}$                                                  & Triangle with nodes $u$, $v$ and $w$ \\[0.05cm]
			$deg(u)$										  & Degree of node $u$ \\[0.05cm]
			$e^{(t)}$                                           & Edge on the stream at time $t$ \\[0.05cm]
			$\Sigma = \{e^{(1)}, ..., e^{(m)} \}$                          & Arbitrary order graph stream \\[0.05cm]
			$\mathcal{T}^{(t)}$                                            & Set of global triangles in $G^{(t)}$                                                                    \\                                                                        [0.05cm]
			$T^{(t)}$								& Number of global triangles in $G^{(t)}$: $|\mathcal{T}^{(t)}| = T^{(t)}$\\										 [0.05cm]
			$\mathcal{T}_u^{(t)}$                                          & Set of local triangles in $G^{(t)}$                                                                     for node $u \in V$  \\[0.05cm]
			$T_u^{(t)}$							& Number of local triangles in $G^{(t)}$                                                                     for node $u \in V$: $|\mathcal{T}_u^{(t)}| = T_u^{(t)}$ \\  [0.05cm]
			\midrule
			\multicolumn{2}{l}{\textit{Notations for Algorithms and Analysis}} \\
			\midrule
			$k$                                                            & Memory budget, i.e., maximum number of edges                                                                      that can be stored \\[0.05cm]
			$\mathcal{W}$                                                  & Waiting room \\[0.05cm]
			$\alpha$                                                      & Fraction of the size of $\mathcal{W}$, i.e.,                                                              $|\mathcal{W}|/ k = \alpha, \alpha \in (0,1)$                                                                                   \\[0.05cm] 
			$\mathcal{H}$                                                  & Set of heavy edges \\[0.05cm]
			$\beta$                                                        & Fraction of the size of $\mathcal{H}$,  i.e., $|\mathcal{H}| / k(1 - \alpha) = \beta$, $\beta \in (0, 1)$ \\[0.05cm]
			$\mathcal{S}_L$                                                & Sample set of light edges of size                                                                                 $|\mathcal{S}_L| = s_{\ell} = k(1 - \alpha)(1 - \beta)$                                                                        \\[0.05cm]
			$\mathcal{S} = (\hat{V}, \hat{E})$        
			& Subgraph of stored edges: $\mathcal{S} =                         \mathcal{W} \cup \mathcal{H} \cup \mathcal{S}_L $ of size  $|\mathcal{S}| = k$ \\[0.05cm]
			$O_H$                                                          & Predictor for the measure for heaviness of edges \\[0.05cm]
			$L$                                                  & Set of 
			light edges in the stream, of size                                                                     $\ell = |L|$\\[0.05cm]
			$\hat{\mathcal{N}}_u$                                          & Set of neighbors of node $u$ in                                                                                   $\mathcal{S} $ \\[0.05cm]
			$\hat{T}$                                                      & Estimate of global triangles count \\[0.05cm]
			$\hat{T}_u$                                                    & Estimate of local triangles count                                                                  for node $u \in                                                                                    \hat{V}$ \\[0.05cm]
			$p_{uvw}$													& Probability of counting triangle $\{u, v, w\}$ \\[0.05cm]
			\bottomrule
	
	\end{tabular}}
	\caption{Table of notation and symbols used in our work}
	\label{tab:symbols}
\end{table}

\section{Description of Datasets}
\label{sec:dataset_description}
In this section, we provide a complete description of the datasets considered in Table~\ref{tab:datasets}.
Recall that, from each dataset, we removed self-loops and multiple edges, deriving a stream of insertion-only edges, for consistency with previous works. 
We report links for downloading each dataset in the README file in \url{https://github.com/VandinLab/Tonic}.
\vspace{\baselineskip}

\noindent \textbf{Single graphs:}
\setlist{nolistsep}
\begin{myitemize}
	\item \textit{Edit EN Wikibooks} contains the edit network of the English Wikipedia, representing users and pages connected by edit events. This dataset is also considered in \cite{chen2022triangle};
	\item \textit{SOC Youtube Growth} includes a list of all of the user-to-user links in Youtube video-sharing social network;
	\item \textit{Cit US Patents} \cite{hall2001nber} represents the citation graph between US patents, where each edge $\{u, v\}$ indicates that patent $u$ cited patent $v$ (used also in \cite{shin2020fast});
	\item \textit{Actors Collaborations} contains actors connected by an edge if they both appeared in a same movie. Thus, each edge is one collaboration between actors;
	\item \textit{Stackoverflow} represents interactions from the StackExchange site ``Stackoverflow''. The network is between users, and edges represent three types of interactions: answering a question of another user, commenting on another  user's question, and commenting on another user's answer;
	\item \textit{SOC LiveJournal} is a friendship network from LiveJournal free on-line community.
\end{myitemize}
\vspace{\baselineskip}
\noindent\textbf{Snapshot sequences:}
\begin{myitemize}
	\item \textit{Oregon} is a sequence of 9 graphs of Autonomous Systems (AS) peering information inferred from Oregon route-views between March 31 2001 and May 26 2001;
	\item \textit{AS-CAIDA} contains 122 RouteViews BGP graph snapshots, from January 2004 to November 2007;
	\item \textit{AS-733} are 733 daily instances which span an interval of 785 days from November 8 1997 to  January 2 2000, from the BGP logs.
	\item \textit{Twitter} \cite{boldi2011layered, Kwak10www} comprises 4 graphs of the Twitter following/followers network. In all our experiments, we consider each of the 4  single networks independently, and also the larger network \textit{Twitter-merged} obtained by merging the 4 graphs, used  also in \cite{stefani2017triest}.
\end{myitemize}

\section{Extension to Fully Dynamic Streams}
\label{sec:additionalalgos}
In this section, we describe the adaptation of our algorithm to fully dynamic (FD) graph streams. First, in Section~\ref{sec:random-pairing} we provide some preliminary notions of random pairing technique, used to maintain size-bounded uniform sample of light edges during the FD stream. Then, in Section~\ref{sec:tonc-fd-algo} we provide a description of our algorithm \algnamedel\ leveraging random pairing to adapt to FD scenarios. 

\subsection{Random Pairing}
\label{sec:random-pairing}
In our algorithm \algnamedel\ we sample edges through \emph{random pairing}~\cite{gemulla2008maintaining}, a method for incrementally maintaining a bounded-size uniform random sample of the items in a dataset in the presence of an arbitrary sequence of insertions and deletions, also adopted in~\cite{stefani2017triest, shin2018think, shin2020fast, lee2020temporal}. The goal of random pairing (RP) is to compensate sample deletions using subsequent insertions: if, at any point, all the previous edge deletions have been compensated, we will have maximum size for our sample, that is $|\mathcal{S}_L| = k(1 - \alpha)(1 - \beta) = s_\ell$.
At any time $t$ in the stream, \algnamedel\ stores in $\mathcal{S}_L$ a set of $k' \leq s_\ell$ edges sampled uniformly at random among all the \emph{light edges} of $G^{(t)}$, that are edges of $G^{(t)}$ that are neither in the waiting room $\mathcal{W}^{(t)}$ nor kept in the set $\mathcal{H}^{(t)}$ of (predicted) heavy edges. 
Following RP scheme, \algnamedel\ maintains counters $d_g$ and $d_b$ for respectively the number of good and the number of bad ``uncompensated'' deletions. RP works as follows: when receiving an edge deletions, it removes the edge from the sample $\mathcal{S}_L$ if present, and increments the counter $d_b$ for bad deletions, or otherwise it ignores the deletions and just increment the counter $d_g$ for good deletions. When receiving an edge insertion, if the number of uncompensated deletions $d = d_g + d_b = 0$, it proceeds as in standard Reservoir Sampling (see Section ~\ref{sec:reservoirs}). If $d > 0$, we need to account for the uncompensated past deletion(s): we flip a coin and include the incoming edge insertion with probability $d_b / (d_b + d_g)$, or otherwise we exclude the edge insertion from the sample; accordingly, we decrease counter $d_b$ if the edge is added to the sample to compensate for an old bad deletion, or we decrease the counter $d_g$ if the edge is excluded from the sample to compensate for an old good deletion.

Let $\mathcal{S}_L^{(t)}$  be the set of edges in $\mathcal{S}_L$ at the end of time step $t$. Random Pairing guarantees~\cite{gemulla2008maintaining} that for every time step $t$ with $|L^{(t)}| \ge k'$, if we let $A$ and $B$ be any two subsets of $L^{(t)}$ of size $|A| = |B| = k'$, then $ \mathbb{P}\left[\mathcal{S}_L^{(t)}= A\right] =  \mathbb{P}\left[\mathcal{S}_L^{(t)}= B\right]$, that is, any two samples of the same size $k'$ are equally likely to be produced, so that RP is a uniform sampling scheme.

\subsection{\algnamedel: Counting Triangles with Predictions in Fully Dynamic Streams}
\label{sec:tonc-fd-algo}
In \texttt{\algnamedel} (Alg.~\ref{alg:tonic_fd}), we provide the pseudocode of the adaptation of our algorithm to fully dynamic graph streams. The general workflow of \texttt{\algnamedel} is very similar to the insertion-only version presented in Section~\ref{sec:alg}. Since we are dealing with both edge insertion and deletions, we make use of random pairing (see Section~\ref{sec:random-pairing}).

Our algorithm \algnamedel\ works as follows: first, initializes the sets and the counters (lines~\ref{line:init-fd}, \ref{line:init-counters-fd}), including $d_g$ and $d_b$ that account for the uncompensated deletions. Then, for each edge in the stream (line~\ref{line:for-each-edge-fd}), given the current edge $\{u, v, \eta\}$, with $\eta \in \{+, -\}$, first it checks whether the edge is an edge insertion $\left( \eta = + \right)$ or an edge deletion $\left( \eta = - \right)$. If we have an edge insertion (line~\ref{line:algfd-eta-ins}), we proceed similarly to Alg.~\ref{alg:tonic-ins}, but paying attention to the number $d_g$ and $d_b$ of respectively good and bad deletion in the stream observed so far. As described in Sect.~\ref{sec:random-pairing}, $d_g$ represents the number of uncompensated edge deletions that involved the removal of edges that were not stored in the sample $\mathcal{S}_L$ (i.e., good deletions), while $d_b$ is the number of uncompensated deletions of edges that were currently stored in $\mathcal{S}_L$ (i.e., bad deletions), causing the decrement of the size of our stored sample $\mathcal{S}_L$. If such numbers are compensated (i.e., their sum is equal to zero, line~\ref{line:algfd-reservoir}), \algnamedel\ resumes to standard reservoir sampling (see Sect.~\ref{sec:reservoirs}). Otherwise, we need to account for the uncompensated number of deletions, and sample $\edge{u'}{v'}$ with probability $d_b / (d_b + d_g)$ (line \ref{line:algfd-uncompensated}). Based on such probability, we compensate the number of deletions by decrementing the respective counter (lines \ref{line:algfd-db--},~\ref{line:algfd-dg--}). 
If the current edge in the stream is an edge deletion (line~\ref{line:algfd-eta-del}), \algnamedel\ checks if the edge to remove is currently in the waiting room $\mathcal{W}$ or in the heavy edges set $\mathcal{H}$ (line~\ref{line:algfd-wh-del}); in such case, the algorithm removes the edge from the respective set. Otherwise, we are dealing with a removal of a light edge. We first need to decrease the number $\ell$ of light edges seen in the stream so far (line \ref{line:algfd-l--}), due to the incoming deletion; then, we have the following cases: (i) the current edge is in the current sample $\mathcal{S}_L$ of light edges. Algorithm~\ref{alg:tonic_fd} removes such edge from $\mathcal{S}_L$ (line~\ref{line:algfd-baddel}) and accounts for a bad deletion incrementing $d_b$ counter (line~\ref{line:db++}). (ii) the current edge is not stored in our sample $\mathcal{S}_L$. \algnamedel\ ignores the edge deletion, and accounts for a good deletion incrementing $d_g$ counter (line~\ref{line:algfd-gooddel}).
When returning the final estimates (line \ref{line:algfd-return}), we return always non-negative global and local triangle counts: in this way, we trade for unbiasedness of the algorithm (in case of negative estimates) by achieving better accuracy of approximations.

\begin{algorithm}[htbp]
	\caption{\algnamedel\  $\left(\Sigma, k, \alpha, \beta, O_H \right)$}
	\label{alg:tonic_fd}
	\footnotesize
	\LinesNumbered
	\kwInput{Arbitrary order fully dynamic edge stream $\Sigma = \{e^{(1)}, e^{(2)}, ... \}$; memory budget $k$; fraction of waiting room space $\alpha$; fraction of heavy edges space $\beta$; edge heaviness predictor $O_H$}
	\kwOutput{Estimate of global triangles count $\hat{T}$; 
		estimate of local triangles count $\hat{T_u}$ for each node $u$}
	$\mathcal{W} \longleftarrow \emptyset$; $\mathcal{H} \longleftarrow \emptyset $; $\mathcal{S}_L \longleftarrow \emptyset $\label{line:init-fd}\;
	$ \hat{T} \longleftarrow 0$; $ \ell \longleftarrow 0$; $d_g \longleftarrow 0$; $d_b \longleftarrow 0$\label{line:init-counters-fd}\;
	\For{each edge $e^{(t)} = \{u, v, \eta \}$ in the stream $\Sigma$\label{line:for-each-edge-fd}}{
		$\mathcal{S} \longleftarrow  \mathcal{W} \cup \mathcal{H} \cup \mathcal{S}_L$\;
		\texttt{CountTriangles-FD}$\left(\{u, v, \eta \}, \mathcal{S}, \ell, d_g, d_b\right)$\label{line:alg2counttriangles-fd}\;
		\uIf{$\eta == +$}{ \label{line:algfd-eta-ins}
			\lIf{$\left|\mathcal{W}\right| < k\alpha$}{$\mathcal{W} \longleftarrow \mathcal{W} \cup \{\{u, v\}\}$} \label{line:w-not-full-fd}
			
			\Else{
				$\{x, y\} \longleftarrow $ oldest edge in $\mathcal{W}$\label{line:w-front-fd}\;
				$\mathcal{W} \longleftarrow \mathcal{W} \setminus \{\{x, y\}\} \cup \{\{u, v\}\}$\; \label{line:w-pops-oldest-fd}
				\If{$\left|\mathcal{H}\right| < k \left(1 - \alpha\right) \beta $\label{line:h-not-full-fd}}{$\mathcal{H} \longleftarrow \mathcal{H} \cup \{\{x, y\}\}$\;}\label{line:h-notfull-gets-oldest-fd}
				
				\Else{
					$ \ell \longleftarrow \ell + 1$\label{line:l++-fd}\;
					$ \{u', v'\} \longleftarrow$ lightest edge in $\mathcal{H} $\label{line:get-lightest-H-fd}\;  
					\uIf{$O_H\left(\{x, y\}\right) >  O_H\left(\{u', v'\}\right)$\label{line:heaviness-comparison-fd}}{
						$\mathcal{H} \longleftarrow \mathcal{H} \cup \{\{x, y\}\} \setminus \{ \{u', v'\}\}$\label{line:h-gets-oldest-fd}\;
					}
					\lElse{ $\{u', v'\} \longleftarrow \{x, y\}$} 
					\uIf{$d_g + d_b == 0$}{ \label{line:algfd-reservoir}
						\uIf{$|\mathcal{S}_L| < k\left(1 - \alpha\right)\left(1 - \beta\right)$}{ \label{line:SL-not-full-fd}
							$\mathcal{S}_L \longleftarrow \mathcal{S}_L \cup \{\{u', v'\}\}$\label{line:SL-gets-lightest-fd}\;
						}
						\Else{
							\texttt{SampleLightEdge}$\left( \{u, v\}, \mathcal{S}_L, \ell \right)$\;
						}
					} \uElseIf{FlipBiasedCoin$\left( \frac{d_b}{d_b + d_g}  \right)$ == HEAD}{ \label{line:algfd-uncompensated}
							$\mathcal{S}_L \longleftarrow \mathcal{S}_L \cup \{\{u', v'\}\}$\;
							$d_b \longleftarrow d_b - 1$\; \label{line:algfd-db--}
						}
						\lElse{ $d_g \longleftarrow d_g - 1$} \label{line:algfd-dg--}
					
				}
			}
		} \uElseIf{$\eta == -$}{ \label{line:algfd-eta-del}
			\lIf{$\{u, v\} \in \mathcal{W} \cup \mathcal{H}$}{ \label{line:algfd-wh-del}
				$\mathcal{W} \cup \mathcal{H} \longleftarrow  \mathcal{W} \cup \mathcal{H} \setminus \{\{u, v\}\}$} \label{line:algfd-deletionH}
			\uElse{
				$\ell \longleftarrow \ell - 1$\;\label{line:algfd-l--}
				\uIf{$\{u, v\} \in \mathcal{S}_L$}{
					$\mathcal{S}_L \longleftarrow \mathcal{S}_L \setminus \{ \{u, v\} \}$\; \label{line:algfd-baddel}
					$d_b \longleftarrow d_b + 1$\;\label{line:db++}
				} \uElse{
					$d_g \longleftarrow d_g + 1$\; \label{line:algfd-gooddel}
				}
			}
		}
	}
	
	\Return \texttt{max}(0, $\hat{T}$), \texttt{max}(0, $\hat{T}_u$) for each node $u \in \mathcal{S}$\label{alg:return-fd}\; \label{line:algfd-return}
\end{algorithm}

In \texttt{CountTriangles-FD} (Alg.~\ref{alg:counttriangles-fd}), we report the procedure for counting triangles in the current subgraph in the fully dynamic setting. Such algorithm is slightly different from \texttt{CountTriangles} (Alg.~\ref{alg:counttriangles2}) since it needs to account for the counters $d_g$ and $d_b$ of good and bad deletions respectively for computing the correction probability, and the type $\eta$ (i.e., edge insertion or deletion) of the current edge in order to increment or decrement the global and local triangle estimates.

\begin{algorithm}[htbp]
\caption{\texttt{CountTriangles-FD}$\left(\{u, v, \eta \}, \mathcal{S}, \ell \right)$}
\label{alg:counttriangles-fd}
\footnotesize
\LinesNumbered
\kwInput{edge $\{u, v, \eta \}$; subgraph $\mathcal{S} = \left(\hat{V}, \ \hat{E}\right)$; number $\ell$ of predicted light edges in the stream $\Sigma^{(t)}$; counter $d_g$ for uncompensated good deletions; counter $d_b$ for uncompensated bad deletions}
\For{each node $w$ in $\hat{\mathcal{N}}_u \cap \hat{\mathcal{N}}_v$}{ \label{line:common-neighs-fd}
	initialize $\hat{T_u}$, $\hat{T_v}$, $\hat{T_w}$ to zero if not set yet\;
	$p_{uvw} = 1$\;
	\uIf{$\{w, u\} \in \mathcal{S}_L$ AND $\{v, w\} \in \mathcal{S}_L$}{
		$p_{uvw} = \min \left( 1, \ \frac{k(1 - \alpha)(1 - \beta)}{\ell + d_g + d_b} \times \frac{k(1 - \alpha)(1 - \beta) - 1}{\ell + d_g + d_b - 1} \right)$\;
	} \uElseIf{$\{w, u\} \in \mathcal{S}_L$ OR $\{v, w\} \in \mathcal{S}_L$}{
		$p_{uvw} = \min \left( 1, \frac{k(1 - \alpha)(1 - \beta)}{\ell + d_g + d_b} \right) $\;
	}
	\uIf{$\eta == +$}{
		increment $\hat{T}$, $\hat{T}_u$, $\hat{T}_v$, $\hat{T}_w$ by $ 1 / p_{uvw}$\; 
	} \uElseIf{$\eta == -$}{
		decrement $\hat{T}$, $\hat{T}_u$, $\hat{T}_v$, $\hat{T}_w$ by $ 1 / p_{uvw}$\; 
	}
}

\end{algorithm}

\section{Proofs}
\label{sec:proofs}
In this section, we provide theoretical proofs of our claims. We describe the correctness (probabilities computation), unbiasedness, variance, time and space complexity of our algorithms \algname\ and \algnamedel. Also, we prove the propositions in Sect.~\ref{sec:analysis_variance_compariso} on the variance comparison between \algname\ algorithm and \wrs\ algorithm.

\subsection{Probabilities Computation}
\label{sec:prob_comp}

In the following, we prove that the probability $p_{uvw}$ computed and used by procedure \texttt{CountTriangles}, Alg.~\ref{alg:counttriangles2} (resp. \texttt{CountTriangles-FD}, Alg.~\ref{alg:counttriangles-fd}) within \algname\  (\algnamedel), corresponds to the probability that triangle $\{u,v,w\}$ is counted by \algname\ (counted or deleted by \algnamedel).

In general,  given a triangle $\{u, v, w\}$ discovered when edge $\{u,v\}$ is observed in the stream, the corresponding probability $p_{uvw}$  depends on where the other two edges $\{v, w\}, \{w,u\}$ are stored by our algorithm, the counter $\ell$ of predicted light edges seen so far, and on the number $d_g$ of good and the number $d_b$ of bad deletions if $\{u,v\}$ is an edge deletion.  Let us denote without loss of generality $e_{uvw}^{(1)} = \{v, w\}$ as the first edge of the triangle $\triang{u}{v}{w}$ arrived in the stream; similarly $e_{uvw}^{(2)}~=~\{w, u\}$ as the second one, and $e_{uvw}^{(3)} = \{u, v\}$ as the last one, which corresponds to the edge currently on the stream at the moment in which we discover the triangle $\{u, v, w\}$. We have the following results.

\begin{lemma}
	\label{lemma:triangleprob2}
	Let $\{u,v\}$ be the edge observed at time $t$ in the \emph{insertion-only} stream and let $\ell$ be the number of light edges seen up to time $t$ in the stream. For any triangle $\{u,v,w\}$ with last edge $\{u,v\}$, the probability $p_{uvw}$ that $\{u,v,w\}$ is discovered by \algname\ is:
	\begin{enumerate}
		\item if $\{v, w\} \text{ and }  \{w, u\} \in \mathcal{W}^{(t - 1)} \cup \mathcal{H}^{(t - 1)} $, then $p_{uvw} = 1$;
		\item if $\{v, w\} \text{ and }  \{w, u\} \in \mathcal{S}_L^{(t - 1)}$, \sloppy{then $p_{uvw}~=~\min\left\{ 1, \ \frac{k(1 - \alpha)(1 - \beta)}{\ell} \times \frac{k(1 - \alpha)(1 - \beta) - 1}{\ell - 1} \right\}$;}		
		\item if only one between $\{v, w\} \text{ and }  \{w, u\}$ is in $\mathcal{W}^{(t - 1)}~\cup~\mathcal{H}^{(t - 1)} $, then $p_{uvw} = \min\left\{ 1, \ \frac{k(1 - \alpha)(1 - \beta)}{\ell} \right\}$;
	\end{enumerate}
\end{lemma}
\begin{proof}
	Observe that when the edge $e^{(t)} = \{u, v\}$ arrives, the triangle $\{u, v, w\}$ is counted if and only if $\{v, w\}$ and $\{w, u\}$ are in $\mathcal{H}^{(t-1)} \cup \mathcal{W}^{(t-1)} \cup \mathcal{S}^{(t - 1)}_L = \mathcal{S}^{(t - 1)}$.
	Proceeding by cases, we have:
	\begin{enumerate}
		\item both edges in the waiting room or heavy edge set, i.e., $\{w, u\} \ \in \ \mathcal{W}^{(t - 1)} \cup \mathcal{H}^{(t - 1)}$ and $\{v, w\} \ \in~\mathcal{W}^{(t - 1)}~\cup~\mathcal{H}^{(t - 1)}$. Since \algname\ (Alg.~\ref{alg:tonic-ins}) always stores edges in the waiting room and in the heavy edge set (lines~\ref{line:alg2WR1},~\ref{line:alg2WR2},~\ref{line:alg2Hnotfull},~\ref{line:alg2H}), then the triangle is counted with probability $p_{uvw}^{(t)} = 1$;
		\item both edges $\{v, w\} \text{ and }  \{w, u\}$ are in the sample of light edges $\mathcal{S}_L^{(t - 1)}$: here, we need to distinguish between two cases. If the sample of light edges $\mathcal{S}_L^{(t - 1)}$ is not full yet, i.e., if $\ell \leq k(1 - \alpha)(1 - \beta)$, we have sampled both edges deterministically, hence the discovered triangle is counted with probability 1. Otherwise, when  $\ell > k(1 - \alpha)(1 - \beta)$, we have that:
		\begin{small}
			\begin{align*}
				& \mathbb{P}\left[\{v, w\} \ \in \mathcal{S}_L^{(t - 1)} \ \text{and} \ \{w, u\} \ \in \mathcal{S}_L^{(t - 1)}\right] = \\
				&= \mathbb{P}\left[\{v, w\} \ \in \mathcal{S}_L^{(t - 1)}\right] \\
				& \times \ \mathbb{P}\left[\{w, u\} \ \in \mathcal{S}_L^{(t - 1)} \ | \ \{v, w\} \ \in \mathcal{S}_L^{(t - 1)}\right] = \\
				&= \frac{k(1 - \alpha)(1 - \beta)}{\ell} \times \ \frac{k(1 - \alpha)(1 - \beta) - 1}{\ell - 1}  = p_{uvw}^{(t)},
			\end{align*}
		\end{small}
		\noindent hence, the overall probability can be written as $p_{uvw}^{(t)}~=~\min\left\{ 1, \ \frac{k(1 - \alpha)(1 - \beta)}{\ell} \times \frac{k(1 - \alpha)(1 - \beta) - 1}{\ell - 1} \right\}$;
		
		\item if only one between $\{v, w\}$ and  $\{w, u\}$ is in $\mathcal{S}_L^{(t - 1)}$, similarly as before, we count the triangle with probability 1 if the sample of light edges is not full yet, else (assuming, without loss of generality, that $\{v, w\}$ is the edge observed in $\mathcal{S}_L^{(t - 1)}$ at time $t$) we can write:
		\begin{align*}
			& \mathbb{P}\left[\{v, w\} \ \in \mathcal{S}_L^{(t - 1)} \ \text{and} \ \{w, u\} \ \in \mathcal{H}^{(t - 1)} \cup \mathcal{W}^{(t - 1)} \right] = \\
			&= \mathbb{P}\left[\{v, w\} \ \in \mathcal{S}_L^{(t - 1)}\right]= \\
			&= \frac{k(1 - \alpha)(1 - \beta)}{\ell}  = p_{uvw}^{(t)},  
		\end{align*}
		\noindent thus, we have that $p_{uvw}^{(t)} = \min\left\{ 1, \ \frac{k(1 - \alpha)(1 - \beta)}{\ell} \right\}$.
	\end{enumerate}
\end{proof}

\begin{lemma}
	\label{lemma:triangleprobfd}
	Let $\{u,v\}$ be the edge observed at time $t$ in the \emph{fully dynamic} stream, let $\ell$ be the number of light edges seen up to time $t$ in the stream, and let $d_g$ and $d_b$ be respectively the uncompensated number of good and bad deletions up to time $t$ in the stream. For any triangle $\{u,v,w\}$ with last edge $\{u,v\}$, the probability $p_{uvw}$ that $\{u,v,w\}$ is discovered by \algnamedel\ is:
	\begin{enumerate}
		\item if $\{v, w\} \text{ and }  \{w, u\} \in \mathcal{W}^{(t - 1)} \cup \mathcal{H}^{(t - 1)} $, then $p_{uvw} = 1$;
		\item if $\{v, w\} \text{ and }  \{w, u\} \in \mathcal{S}_L^{(t - 1)}$, \sloppy{then $p_{uvw}~=~\min\left\{ 1, \ \frac{k(1 - \alpha)(1 - \beta)}{\ell + d_g + d_b} \times \frac{k(1 - \alpha)(1 - \beta) - 1}{\ell + d_g + d_b - 1} \right\}$;}		
		\item if only one between $\{v, w\} \text{ and } \{w, u\}$ is in $\mathcal{W}^{(t - 1)}~\cup~\mathcal{H}^{(t - 1)} $, then $p_{uvw} = \min\left\{ 1, \ \frac{k(1 - \alpha)(1 - \beta)}{\ell + d_g + d_b} \right\}$;
	\end{enumerate}
\end{lemma}
\begin{proof}
	Observe that when the edge $e^{(t)} = \{u, v, \eta\}$ arrives, the triangle $\{u, v, w\}$ is counted (if $\eta = +$) or deleted (if $\eta = -$)  if and only if $\{v, w\}$ and $\{w, u\}$ are in $\mathcal{H}^{(t-1)}~\cup~\mathcal{W}^{(t-1)}~\cup~\mathcal{S}^{(t - 1)}_L = \mathcal{S}^{(t - 1)}$.
	We have:
	\begin{enumerate}
		\item both edges in the waiting room or heavy edge set, i.e., $\{w, u\} \ \in \ \mathcal{W}^{(t - 1)} \cup \mathcal{H}^{(t - 1)}$ and $\{v, w\} \ \in~\mathcal{W}^{(t - 1)}~\cup~\mathcal{H}^{(t - 1)}$. Since \algnamedel\ (Alg.~\ref{alg:tonic_fd}) always stores edges in the waiting room and in the heavy edge set (lines~\ref{line:w-not-full-fd},~\ref{line:w-pops-oldest-fd},~\ref{line:h-notfull-gets-oldest-fd},~\ref{line:h-gets-oldest-fd}), then the triangle is counted or deleted with probability $p_{uvw}^{(t)} = 1$;
		\item both edges $\{v, w\} \text{ and } \{w, u\}$ are in the sample of light edges $\mathcal{S}_L^{(t - 1)}$: here, we need to distinguish between two cases. If the sample of light edges $\mathcal{S}_L^{(t - 1)}$ is not full yet, i.e., if $\ell \le |\mathcal{S}_L^{(t)}|$, we have sampled both edges deterministically, hence the discovered triangle is counted or deleted with probability 1. Otherwise, when  $\ell > |\mathcal{S}_L^{(t)}|$, we have that:
		\begin{small}
			\begin{align*}
				& \mathbb{P}\left[\{v, w\} \ \in \mathcal{S}_L^{(t - 1)} \ \text{and} \ \{w, u\} \ \in \mathcal{S}_L^{(t - 1)}\right] = \\
				&= \mathbb{P}\left[\{v, w\} \ \in \mathcal{S}_L^{(t - 1)}\right] \\
				& \times \ \mathbb{P}\left[\{w, u\} \ \in \mathcal{S}_L^{(t - 1)} \ | \ \{v, w\} \ \in \mathcal{S}_L^{(t - 1)}\right] = \\
				&= \frac{k(1 - \alpha)(1 - \beta)}{\ell + d_g + d_b} \times \ \frac{k(1 - \alpha)(1 - \beta) - 1}{\ell + d_g + d_b - 1}  = p_{uvw}^{(t)},
			\end{align*}
		\end{small}
		hence, the overall probability can be written as $p_{uvw}^{(t)}~=~\min\left\{ 1, \ \frac{k(1 - \alpha)(1 - \beta)}{\ell + d_g + d_b} \times \frac{k(1 - \alpha)(1 - \beta) - 1}{\ell + d_b + d_g - 1} \right\}$;
		
		\item if only one between $\{v, w\}$ and  $\{w, u\}$ is in $\mathcal{S}_L^{(t - 1)}$, similarly as before, we count or delete the triangle with probability 1 if the sample of light edges is not full yet, else (assuming, without loss of generality, that $\{v, w\}$ is the edge observed in $\mathcal{S}_L^{(t - 1)}$ at time $t$) we can write:
		\begin{align*}
			& \mathbb{P}\left[\{v, w\} \ \in \mathcal{S}_L^{(t - 1)} \ \text{and} \ \{w, u\} \ \in \mathcal{H}^{(t - 1)} \cup \mathcal{W}^{(t - 1)} \right] = \\
			&= \mathbb{P}\left[\{v, w\} \ \in \mathcal{S}_L^{(t - 1)}\right]= \\
			&= \frac{k(1 - \alpha)(1 - \beta)}{\ell + d_b + d_g}  = p_{uvw}^{(t)},
		\end{align*}
		thus, we have that $p_{uvw}^{(t)} = \min\left\{ 1, \ \frac{k(1 - \alpha)(1 - \beta)}{\ell + d_b + d_g} \right\}$.
	\end{enumerate}
\end{proof}

\subsection{Unbiasedness of \algname}

We now prove that \algname\ and \algnamedel\ report unbiased estimates of the true global and local triangle counts (Thm~\ref{thm:unbiasedness}). In our proof we assume that $|\Sigma|>3$.

\begin{theorem}
	\algname\ and \algnamedel\ return unbiased estimates of the global triangle count and of the local triangle counts for each node, at any time $t$. That is, if we let $T^{(t)}$ and $T_u^{(t)}$ be respectively the true global count of triangles in the graph and the true local triangle count for node $u \in V$ at time $t$, we have:
	\begin{equation}
		\mathbb{E}\left[\hat{T}^{(t)}\right] = T^{(t)}, \ \forall \ t \geq 0
		\label{eq:gloablcount2-appendix}
	\end{equation}
	\begin{equation}
		\mathbb{E}\left[\hat{T}_u^{(t)}\right] = T_u^{(t)} \ \forall u \in V, \ \forall \ t \geq 0
		\label{eq:localcount2-appendix}
	\end{equation}
	\label{theorem:unbiasedness2-appendix}
\end{theorem}

\begin{proof}
	Let $\delta_{uvw}^{(s)}$ be the random variable representing the amount of increase or decrease of the counters due to the discovery or deletion of the triangle $\{u, v, w\}^{(s)}$ at time $s \leq t$, given the current edge $\edget{u}{v}{+}^{(t)}$ for \algname\ (insertion-only), and $\edget{u}{v}{\eta}^{(t)}$ for \algnamedel. 
	Then, the following holds:
	
	\begin{small}
	\begin{equation}
		\delta_{uvw}^{(s)} = 
		\begin{cases}
			1 / p_{uvw}^{(s)} & \text{if } \eta = + \text{ and } \{v, w \}, \{w, u\} \in \mathcal{S}^{(s - 1)}\\
			- 1 / p_{uvw}^{(s)} & \text{if } \eta = - \text{ and }  \{v, w \}, \{w, u\} \in \mathcal{S}^{(s - 1)}\\
			0                 & \text{otherwise.}
		\end{cases}
		\label{eq:delta_uvw2}
	\end{equation}
	\end{small}
	Combining the above with Lemma~\ref{lemma:triangleprob2}~and~\ref{lemma:triangleprobfd}, we have $\mathbb{E}\left[\delta_{uvw}^{(s)}\right] = 1 / p_{uvw}^{(s)} \times p_{uvw}^{(s)} \ + \ 0 \ \times (1 - p_{uvw}^{(s)}) = 1$ if the triangle $\triang{u}{v}{w}^{(s)}$ is counted, or $\mathbb{E}\left[\delta_{uvw}^{(s)}\right] = - 1 / p_{uvw}^{(s)} \times p_{uvw}^{(s)} \ + \ 0 \ \times (1 - p_{uvw}^{(s)}) = -1$ if the triangle $\triang{u}{v}{w}^{(s)}$ is deleted.
	We can express the estimated number of triangles $\hat{T}^{(t)}$ as $\hat{T}^{(t)} = \sum_{\{u, v, w\}^{(s)} \in \mathcal{T}^{(t)}} \ \delta_{uvw}^{(s)}$, for $s \leq t$. So, we can write:
	\begin{small}
	\begin{align*}
		\mathbb{E}\left[\hat{T}^{(t)}\right] =& \ \mathbb{E} \left[ \sum_{\{u, v, w\}^{(s)} \in \mathcal{T}^{(t)}} \ \delta_{uvw}^{(s)} \right] = \\
		= &  \sum_{\{u, v, w\}^{(s)} \in \mathcal{T}^{(t)}} \ \mathbb{E}\left[\delta_{uvw}^{(s)}\right] = \left|\mathcal{T}^{(t)}\right| = T^{(t)} \\
	\end{align*}
	\end{small}
	which proves Eq. \ref{eq:gloablcount2-appendix}. Similarly, we can derive the estimated local count for each node $u$ as follows:
	\begin{small}
	\begin{align*}
		\mathbb{E}\left[\hat{T}_u^{(t)}\right] =& \ \mathbb{E} \left[ \sum_{\{u, v, w\}^{(s)} \in \mathcal{T}_u^{(t)}} \ \delta_{uvw}^{(s)} \right] = \\
		= & \sum_{\{u, v, w\}^{(s)} \in \mathcal{T}_u^{(t)}} \ \mathbb{E}\left[\delta_{uvw}^{(s)}\right] = \left|\mathcal{T}_u^{(t)}\right| = T_u^{(t)}
	\end{align*}
	\end{small}
	which proves Eq. \ref{eq:localcount2-appendix}.
\end{proof}

\subsection{Time Complexity}

We now prove the following theorem that establishes the time complexity of \algname\ and \algnamedel.
\begin{manualtheorem}{Theorem~\ref{theorem:timecomplexity2}} 
	Given an input graph stream $\Sigma$, and given $\alpha = \left|\mathcal{W}\right|/k$, $\beta = \left|\mathcal{H}\right|/k(1 - \alpha)$, \algname\ and \algnamedel\ with memory budget $k$ process each edge in the stream $\Sigma$ in $\bigO{\left( k + \log \left(k (1 - \alpha) \beta \right) \right)}$ time.
	\label{theorem:timecomplexity2-appendix}
\end{manualtheorem}

\begin{proof}
	For each incoming edge $\{u, v, \eta\}$ in the stream $\Sigma$, the most expensive steps are the computation of common neighbors (line \ref{line:alg2commonneighs} of Alg.~\ref{alg:counttriangles2},  line \ref{line:common-neighs-fd} of Alg.~\ref{alg:counttriangles-fd}) and the insertion, retrieval or deletion of the lightest edge in $\mathcal{H}$ (line~\ref{line:alg2lightestHE} of Alg.~\ref{alg:tonic-ins}, line~\ref{line:get-lightest-H-fd} of Alg.~\ref{alg:tonic_fd}). Also, for \algnamedel\, such cost is paid if the deletion of current edge occurs in the heavy edge set $\mathcal{H}$ (line~\ref{line:algfd-deletionH}). In general, we are assuming constant time for: operations in the waiting room $\mathcal{W}$ (implemented through a FIFO queue in \algname, and through an indexed hash set in \algnamedel), for querying the predictor $O_H$, for the coin flip and for accessing, adding or removing an element in the set of light edges $\mathcal{S}_L$, implemented through an array with fixed size $s_{\ell}$. \algname\ and \algnamedel\ take $\bigO{\left(k\right)}$ to perform the computation of common neighbors.
	Then, we need to account for the operations on $\mathcal{H}$; assume that such set is implemented through a min-heap priority queue, where insertion and deletion takes $\bigO{\left( \log(|\mathcal{H}|) \right)}$. Thus, \algname\ and \algnamedel\ process each incoming edge in $\bigO{\left( k + \log \left(k (1 - \alpha) \beta \right) \right)}$ time.
\end{proof}

\subsection{Space Complexity}
For the space complexity, \algname\ and \algnamedel\ do not exceed the given fixed memory $k$ while storing the edges of the stream. The proof of the related proposition below is trivial.

\begin{proposition}
	\label{theorem:spacecomplexity2-appendix}
	Given an input graph stream $\Sigma$, at the end of the stream, \algname\ and \algnamedel\ store $\bigO{(k)}$ edges to compute the global and local estimates of triangles for the entire graph stream.
\end{proposition}

Note that the above analysis does not consider the space required to store the local counters $\hat{T}_u$ for nodes $u$, which requires an additional $\bigO{\left( \hat{V}^{(t)} \right)}$ space, at any time $t$ in the stream. Also, note that \algname\ and \algnamedel\ explicitly represent only counters for nodes that have at least one adjacent triangles (i.e., only counters $>0$).

\subsection{\algname\ vs \wrs: Variance Comparison} 

We now prove Prop.~\ref{prop:var_vs_WRS}. We recall that, for the sake of simplicity, in the analysis we consider a simplified version of \wrs\ algorithm and a simplified version of \algname\ algorithm, considering \emph{insertion-only} graph streams. The simplified version of the \wrs\ samples each light edge (i.e, that leaves the waiting room) independently with probability $p$;  the simplified version of \algname\  uses an oracle that predicts an edge leaving the waiting room as heavy or light, keeps (predicted) heavy edges in $\mathcal{H}$, and samples each light edge with probability $p' <p$, where $p$ is the probability that an edge is sampled by \wrs. Note that by properly fixing $p'$, according to the memory budget for heavy edges $\left( k \left(1 - \alpha\right)\beta \right)$ the two algorithms have the same expected memory budget. 
We assume that the waiting room has the same size in both \wrs\ and \algname. 

Also, we will assume heavy edges appear in~$\ge~\rho$ triangles, while light edges appear in~$< \rho$ triangles, for some $\rho \ge 3$. However, to capture the errors of the predictor, in particular around the threshold $\rho$, we assume that edges $h$ predicted as heavy by the predictor are involved in $\Delta(h)\ge \rho/c$ triangles, while edges $l$ predicted as light by the predictor are involved in $\Delta(l) < c\rho$ triangles, for some constant $c \ge 1$.  Note that for edges $e$ with $\rho / c < \Delta(e) < c\rho$ the oracle can make arbitrarily wrong predictions.
We now prove the following result that provides a condition under which the variance of the estimate of global triangle counts reported by \algname\ is smaller than the estimate of \wrs.

\begin{manualtheorem}{Proposition~\ref{prop:var_vs_WRS}}
	\label{prop:var_vs_WRS-appendix}
	Let $\Var[\hat{T}_{\wrs}(p)]$ be the variance of the estimate  $\hat{T}_{\wrs}$ obtained by \wrs\ when edges are sampled independently with probability $p$, and let  $\Var[\hat{T}_{\algname}(p')]$ be the variance of the estimate  $\hat{T}_{\algname}$ obtained by \algname\ when \emph{light} edges are sampled with probability $p'$. Then $\Var[\hat{T}_{\wrs}(p)] \ge \Var[\hat{T}_{\algname}(p')]$ if 
	\begin{equation*}
		\frac{T^H}{T^L} > 3 \frac{(1/p'^2 - 1/p^2)+c\rho(1/p' - 1/p)}{(1/p-1)(3+4\rho/c)}.
	\end{equation*}
\end{manualtheorem}

\begin{proof}
	Let $\hat{T}_{algo}$ be the number of triangles estimated by \textit{algo},  where \textit{algo} $\in \{$\wrs, \algname$\}$, and $\mathcal{T}^{S_1, S_2}$ the set of triangles that, when the last edge closing a triangle in such set is observed in the stream, the other two edges are found in $S_1 \text{ and } S_2$, where $S_1, S_2 \in  \{\mathcal{H}, \mathcal{W}, L\}$. Also, let $T^{S_1, S_2}$, $\hat{T}^{S_1, S_2}$ be respectively the true and estimated number of triangles in $\mathcal{T}^{S_1, S_2}$. Without loss of generality, we assume that $\mathcal{T}^{S_1, S_2} = \mathcal{T}^{S_2, S_1}$. We recall that triangles counted surely, that is with probability 1, have zero variance (these are triangles in $\mathcal{T}^{\mathcal{W}, \mathcal{W}}$ for \wrs, and in $\mathcal{T}^{\mathcal{W}, \mathcal{W}}, \mathcal{T}^{\mathcal{H}, \mathcal{H}},\mathcal{T}^{\mathcal{W}, \mathcal{H}}$ for \algname). For the other triangles, we have:
	\begin{small}
	\begin{align*}
		& \Var \left[ \hat{T}^{S_1, S_2}_{algo} \right] = \Var \left[ \sum_{\{u, v, w\} \in \mathcal{T}_{algo}^{S_1, S_2}} \delta_{\{u, v, w\}} \right] \\
		&= \sum_{\{u, v, w\} \in \mathcal{T}_{algo}^{S_1, S_2}} \Var \left[ \delta_{\{u, v, w\}} \right] + \Cov \left[ \hat{T}^{S_1, S_2}_{algo} , \hat{T}^{S_1, S_2}_{algo} \right] \\
		&= \sum_{\{u, v, w\} \in \mathcal{T}_{algo}^{S_1, S_2}} \left( \mathbb{E} \left[ \delta_{\{u, v, w\}}^2 \right] - \left(\mathbb{E} \left[ \delta_{\{u, v, w\}} \right] \right)^2  \right) + \\
		& \ \ \ \ \ + \Cov \left[ \hat{T}^{S_1, S_2}_{algo} , \hat{T}^{S_1, S_2}_{algo} \right] \\
		&= T_{algo}^{S_1, S_2} \left( 1/q^2 - 1\right) + \Cov \left[ \hat{T}^{S_1, S_2}_{algo} , \hat{T}^{S_1, S_2}_{algo} \right],
	\end{align*}
	\end{small}
	where $ \delta_{\{u, v, w\}}$ is defined as Eq. ~\ref{eq:delta_uvw2} (but in a fixed probability context), and $q \in \{p, p^2, p', {p'}^2\}$ depending on the algorithm for which we are computing the variance and on the set of triangles $\mathcal{T}^{S_1, S_2}$ we are considering. Note that, for \wrs\ we have $q = p$ for triangles in $\mathcal{T}^{\mathcal{H}, \mathcal{W}} \cup \mathcal{T}^{L, \mathcal{W}}$ and $q = p^2$ for triangles in $\mathcal{T}^{\mathcal{H}, \mathcal{H}} \cup \mathcal{T}^{L, L} \cup \mathcal{T}^{\mathcal{H}, L}$; for \algname, $q = p'$ for triangles in $\mathcal{T}^{\mathcal{H}, L} \cup \mathcal{T}^{\mathcal{W}, L}$ and $q = {p'}^2$ for triangles in $\mathcal{T}^{L, L}$.
	Thus, for \wrs\ we can write:
	\begin{small}
	\begin{align*}
		& \Var\left[\hat{T}_{\wrs}(p)\right] = \\
		& = \Var\left[\hat{T}_{\wrs}^{\mathcal{H}, \mathcal{W}} + \hat{T}_{\wrs}^{L, \mathcal{W}} + \hat{T}_{\wrs}^{\mathcal{H}, \mathcal{H}} + \hat{T}_{\wrs}^{L, L} + \hat{T}_{\wrs}^{\mathcal{H}, L}\right] \\
		& = \left (T^{\mathcal{H}, \mathcal{W}} + T^{L, \mathcal{W}}  \right) (1/p - 1) + \\
		& + \left( T^{\mathcal{W}, L} + T^{L, L} + T^{\mathcal{H}, L}  \right) (1/p^2 - 1) + \\
		& + \sum_{x, y, z, w \in [\mathcal{W}, \mathcal{H},  L]} \Cov \left[ \hat{T}^{x, y} , \hat{T}^{w, z} \right].\\
	\end{align*}
	\end{small}
	Similarly, for \algname\ we have:
	\begin{small}
	\begin{align*}
		\Var\left[\hat{T}_{\algname}(p')\right] &= \Var\left[\hat{T}_{\algname}^{L, \mathcal{W}} + \hat{T}_{\algname}^{\mathcal{H}, L} + \hat{T}_{\algname}^{L, L} \right] \\
		&= \left (T^{\mathcal{H}, L} + T^{L, \mathcal{W}}  \right) (1/p' - 1) + \\
		& + T^{L, L}  (1/{p'}^2 - 1) + \\
		& + \sum_{x, y, z, w \in [\mathcal{W}, \mathcal{H},  L]} \Cov \left[ \hat{T}^{x, y} , \hat{T}^{w, z} \right],\\
	\end{align*}
	\end{small}
	where the last sum is intended over all combinations of the sets in $[\mathcal{W}, \mathcal{H}, L]$.
	
	We can express the covariance terms as:
	\begin{small}
	\begin{align*}
		& \Cov \left[ \hat{T}^{x, y}, \hat{T}^{w, z} \right] = \sum_{\Delta_i \in \mathcal{T}^{x, y}, \Delta_j \in \mathcal{T}^{w, z}}\Cov \left[ \delta_i^{x, y}, \delta_j^{w, z} \right] \nonumber \\
		&= \sum_{\Delta_i \in \mathcal{T}^{x, y}, \Delta_j \in \mathcal{T}^{w, z}}  \mathbb{E} \left[ \delta_i^{x, y}, \delta_j^{w, z} \right] - \mathbb{E} \left[ \delta_i^{x, y} \right] \mathbb{E} \left[ \delta_j^{w, z} \right] \nonumber \\
		& = \sum_{\substack{ \Delta_i, \Delta_j \in \\ \mathcal{T}^{x, y} \cap \mathcal{T}^{w, z} }}  \frac{1}{p_i} \ \frac{1}{p_j} \ \mathbb{P}\left[ \Delta_i, \Delta_j \in \hat{T} \text{ \footnotesize{and} } \Delta_i \cap \Delta_j \text{ \footnotesize{sampled}} \right] - 1  \nonumber \\
		& = \sum_{\Delta_i, \Delta_j \in \mathcal{T}^{x, y} \cap \mathcal{T}^{w, z}}  \frac{1}{p_i} \ \frac{1}{p_j} \ \frac{p_i \ p_j}{q} - 1 \nonumber \\
		& = \left| \mathcal{T}^{x, y} \times \mathcal{T}^{w, z} \right| \left( \frac{1}{q} - 1 \right), \nonumber \\
	\end{align*}
	\end{small}
	where $\delta_i$, $\delta_j$ are respectively the random variable corresponding to the increment in count due to triangle $\Delta_i$ and $\Delta_j$; $\Delta_i \cap \Delta_j$ corresponds to the shared edge between triangle $\Delta_i$ and triangle $\Delta_j$, $q \in \{p, p'\}$ is the sampling probability of the considered algorithm. Note that only triangles having the shared edge that has been sampled (with probability $< 1$) have impact on the covariance term. Thus, in order to have non null covariance terms, we need to have $\Delta_i \cap \Delta_j$ that has been sampled as light edge by the respective algorithm.
	
	Therefore, the difference between the variance of $\hat{T}_{\wrs}(p)$ and the variance of $\hat{T}_{\algname}(p')$ can be expressed as:
	\begin{small}
	\begin{align}
		& \Var[\hat{T}_{\wrs}(p)]-\Var[\hat{T}_{\algname}(p')] = T^{\mathcal{H},\mathcal{W}} (1/p-1) \nonumber \\
		& + T^{\mathcal{H},\mathcal{H}} (1/p^2-1) + T^{\mathcal{H},L} (1/p^2 - 1/p') + 2 S_\mathcal{H} (1/p-1)  \nonumber \\
		& + T^{L,\mathcal{W}} (1/p - 1/p')+ T^{L, L} (1/p^2-1/p'^2) +\nonumber \\
		& + 2 \bar{S}_{\mathcal{H}}(1/p - 1/p').
		\label{eq:diff_variances-appendix}
	\end{align}
	\end{small}
	In Eq. ~\ref{eq:diff_variances-appendix}, we express all the possible combinations of the covariance terms with $S_\mathcal{H} \text{ and } \bar{S}_{\mathcal{H}}$ notation. More explicitly, $S_{\mathcal{H}}$ indicates the number of all pairs of triangles (that share an heavy edge) including at least one deterministic triangle for \algname\ (i.e., having positive covariance in \wrs\ and 0 covariance in \algname), while $\bar{S}_{\mathcal{H}}$ is the number of all pairs of triangles (that share a sampled edge) without deterministic triangles in such pair (i..e., having positive covariance both in \wrs\ and in \algname). Furthermore, such pairs in covariance terms in Eq. ~\ref{eq:diff_variances-appendix} are uniquely counted, hence we need to multiply the cardinality of both sets by 2.
	Note that, assuming that $p' \approx p$ and since $p' <p$, the first four terms are $\ge 0$, while the last three are $\le 0$.  Let $h  \in \mathcal{T}^{\mathcal{H},\mathcal{W}}$ denote an heavy edge: note that $|\mathcal{T}^{\mathcal{H},\mathcal{W}}|~(1/p~-~1)~=~\sum_{h \in \mathcal{T}^{\mathcal{H},\mathcal{W}}} (1/p-1)$, and similarly for all other terms not related to $S_{\mathcal{H}}$ and to $\bar{S}_{\mathcal{H}}$. Moreover, note that:
	\begin{equation*}
		S_{\mathcal{H}} (1/p-1) = \sum_{h \in \mathcal{H}} {\Delta(h) \choose 2} (1/p - 1),
	\end{equation*}
	and 
	\begin{equation*}
		\bar{S}_{\mathcal{H}}(1/p - 1/p') = \sum_{l \in L} {\Delta(l) \choose 2} (1/p - 1/p').
	\end{equation*}
	We have:
	\begin{equation*}
		\frac{\Delta(e)^2}{3}\le {\Delta(e) \choose 2} \le \frac{\Delta(e)^2}{2},
	\end{equation*}
	where the first inequality holds for $\Delta(e) \ge 3$, while the second one is always correct (by considering ${r \choose 2} = 0$ if $r \le 1$). The first inequality is used for edges $e$ that are predicted as heavy, therefore requiring $\Delta(e) \ge 3$ is reasonable.
	We can bound:
	\begin{align*}
		\Var[\hat{T}_{\wrs}(p)]-\Var[\hat{T}_{\algname}(p')] \ge
		\sum_{h \in \mathcal{T}^{\mathcal{H},\mathcal{W}}} (1/p-1) \\
		+ \frac{1}{2} \sum_{h \in \mathcal{T}^{\mathcal{H},\mathcal{H}}} (1/p^2-1)
		+ \sum_{h \in \mathcal{T}^{\mathcal{H},L}} (1/p^2 - 1/p')\\
		+ 2 \sum_{h \in \mathcal{H}} \frac{\Delta(h)^2}{3} (1/p-1)
		+  \sum_{l \in \mathcal{T}^{L,\mathcal{W}}} (1/p - 1/p') \\
		+ \frac{1}{2}\sum_{l \in \mathcal{T}^{L,L}} (1/p^2-1/p'^2) 
		+ 2 \sum_{l \in L} \frac{\Delta(l)^2}{2}(1/p - 1/p').
	\end{align*}	
	Let $T^H = \sum_{h \in \mathcal{H}} \Delta(h)$. We have: 
	\begin{small}
	\begin{align*}
		\sum_{h \in \mathcal{H}} \frac{\Delta(h)^2}{3} (1/p-1) = 
		\sum_{h \in \mathcal{H}} \Delta(h) \frac{\Delta(h)}{3} (1/p-1) \ge \\
		\frac{1}{3}(1/p-1)\rho/c \sum_{h\in \mathcal{H}} \Delta(h) = \frac{1}{3}(1/p-1)\rho/c T^{{H}}
	\end{align*}
	\end{small}
	\noindent since $\Delta(h) \ge \rho/c$ by the assumptions on the predictor. Analogously,  Let $T^L = \sum_{l \in L} \Delta(l)$. We have:
	\begin{small}
	\begin{align*}
		\sum_{l \in L} \frac{\Delta(l)^2}{2}(1/p - 1/p') = 
		\sum_{l \in L} \Delta(l) \frac{\Delta(l)}{2}(1/p - 1/p') \ge\\
		\frac{1}{2}(1/p - 1/p') c\rho \sum_{l \in L} \Delta(l) = \frac{1}{2}c\rho(1/p-1/p') T^{L}.
	\end{align*}
	\end{small}
	If $p'\approx p$, we have that $1/p - 1 \le 1/p^2 -1/p'$, and $1/p'^2~-~1/p^2~\ge~1/p'~-~1/p^2$. Therefore, we have that:
	\begin{small}
	\begin{align*}
		& \Var[\hat{T}_{\wrs}(p)]-\Var[\hat{T}_{\algname}(p')] \ge \\
		& \ge T^{{H}}\left(\frac{1}{6} (1/p-1)(3+4\rho/c)\right) + \\
		& + T^{L}\left( \frac{1}{2}(1/p^2 - 1/p'^2) +c\rho (1/p -1/p') \right),
	\end{align*}
	\end{small}
	and $\Var[\hat{T}_{\wrs}(p)]-\Var[\hat{T}_{\algname}(p')]  > 0$ if 
	\sloppy{
		$\frac{T^H}{T^L} >~3~\frac{(1/p'^2 - 1/p^2)+c\rho(1/p' - 1/p)}{(1/p-1)(3+4\rho/c)}$}.
\end{proof}

The following result shows that if the predictor $O_H$ provides random predictions, then the variance of the estimate from \algname\ is the same as the variance of the estimate of \wrs.

\begin{manualtheorem}{Proposition~\ref{prop:variance_random_pred}} 
	If the predictor $O_H$ predicts a random set of edges as heavy edges, and \wrs\ and \algname\ use the same amount of memory, then $\Var[\hat{T}_{\wrs}(p)]=\Var[\hat{T}_{\algname}(p')]$.
\end{manualtheorem}
\begin{proof}{(Sketch)}
	If the predictor $O_H$ predicts a random set of edges as heavy edges, then for both \wrs\ and \algname\ the set of edges in memory outside the waiting room is a random subset of light edges. If the two algorithms use the same memory, the two sets have the same cardinality, and the result follows.
\end{proof}

\section{Additional Experimental Results}
\label{sec:additional-experiments}

\subsection{MinDegree Predictor: Relation to Heavy Edges}
\label{sec:mindeg-vs-heaviness}

Fig.~\ref{fig:heaviness-vs-mindeg-appendix} shows the number of triangles (i.e., heaviness) vs the minimum degree of each edge, considering the top 1 \textpertenthousand\ edges sorted by heaviness (for Wikibooks datasets, the top 1000 heaviest edges have been considered). The figure shows that there is a clear and linear relation between the minimum degree of edges and their respective heaviness for the heaviest edge in the graph stream, even if the minimum degree is not a perfect predictor of the heaviness.

\begin{figure}[h]
	\centering
	\includegraphics[width=1 \textwidth]{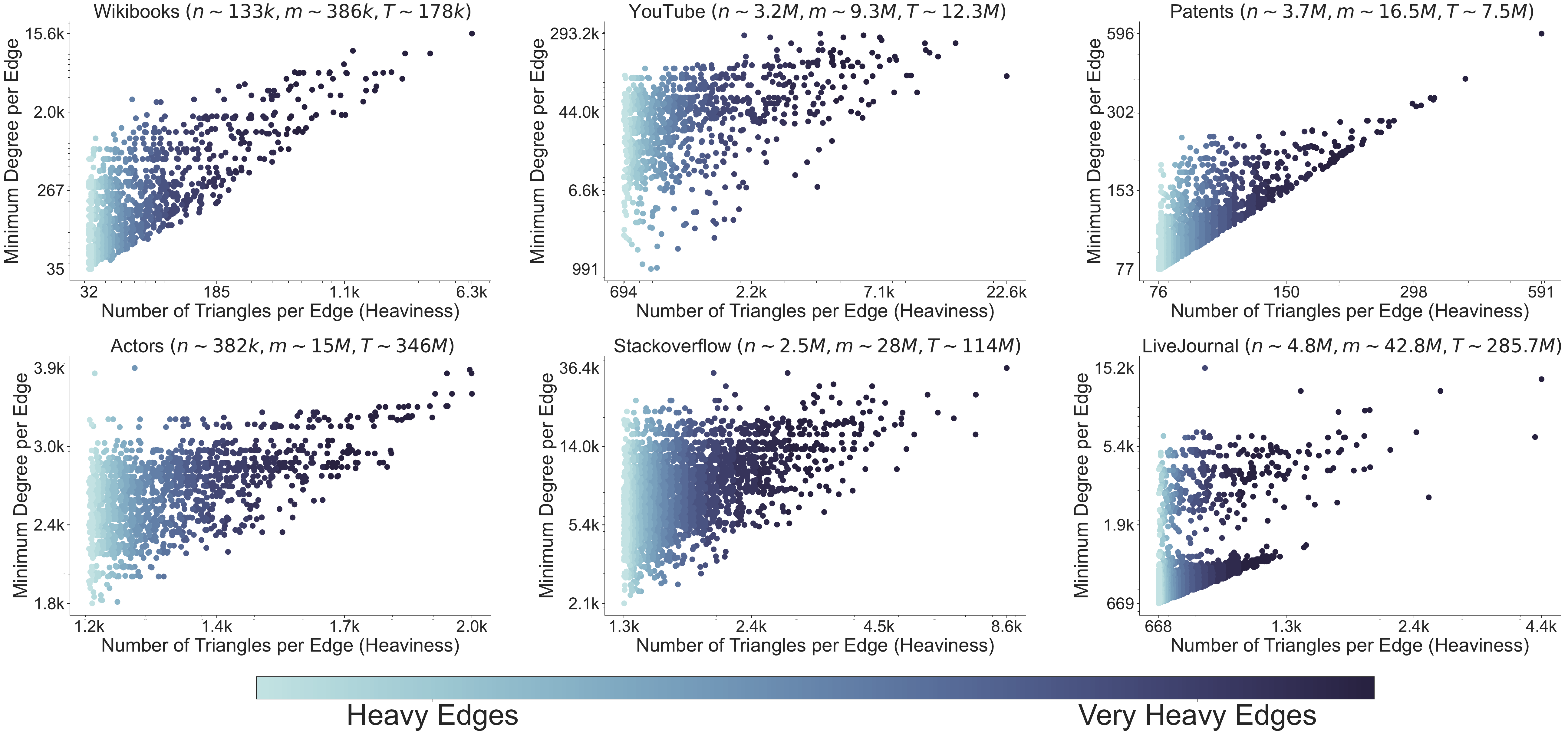}
	\caption{Relation between Number of Triangles (i.e., heaviness) per edge and Minimum Degree per edge through the top 1\textpertenthousand$m$ heaviest edges in the graph stream (for Wikibooks dataset, the top 1000 heaviest edges have been considered). Both axis are represented using a log scale.}
	\label{fig:heaviness-vs-mindeg-appendix}
\end{figure}

\subsection{Accuracy vs Parameters} 
\label{sec:accuracy-vs-params}
The results in Fig.~\ref{fig:accuracy_params_experiments_merged} are obtained by running \wrs , \chenalg,  and \algname\ with values of the main parameters $\alpha$ and $ \beta $ of such algorithms that are:
\begin{align*}
	& \alpha_{\wrs}, \beta_{\chenalg} \in [0.1, 0.15, 0.19, 0.23, 0.24, 0.28, 0.3, 0.36]; \\
	& \alpha_{\algname}, \beta_{\algname} \in [0.05, 0.1, 0.15, 0,2].
\end{align*}
We chose for such values because we were able to obtain the same space for storing edges with probability 1 for each of the considered algorithm, since the deterministic space accounts for a fraction of the memory budget $k$ corresponding to $\alpha_{\wrs}$, $\beta_{\chenalg}$, and $~\alpha_{\algname} + \beta_{\algname} \left(1 - \alpha_{\algname}~\right) $ respectively for \wrs , \chenalg , and \algname. Since for \algname\ some permutation of parameters correspond to the same deterministic space, we report in Fig.~\ref{fig:accuracy_params_experiments_merged} the best of such combinations, besides our supposed choice of parameters $\alpha_{\algname}=~ 0.05 \text{ and } \beta_{\algname}=~0.2$ (highlighted by the dashed circle).

\subsection{Quality of Approximation} 
\label{sec:pdf-experiments}
In the following, we study the quality of the approximation of global triangles returned by \algname\ in terms on unbiasedness and variance. In particular, we report the probability density function of the estimates of the considered algorithms. We ran \wrs, \chenalg, \algname\ for 500 consecutive and independent trials, for each dataset but the Twitter ones, due to time impracticability on such large graphs. Then, we estimated the PDF by using Gaussian Kernel Estimation. The results are shown in Figure~\ref{fig:pdf_merged-appendix}. As expected, all the algorithms return unbiased estimates, and \algname\ is able to achieve a lower variance in all the datasets, leading to more accurate and robust global estimates and confirming our theoretical analysis in Section~\ref{sec:analysis}. 

\begin{figure}[htbp]
	\centering
	\includegraphics[width= 1 \columnwidth]{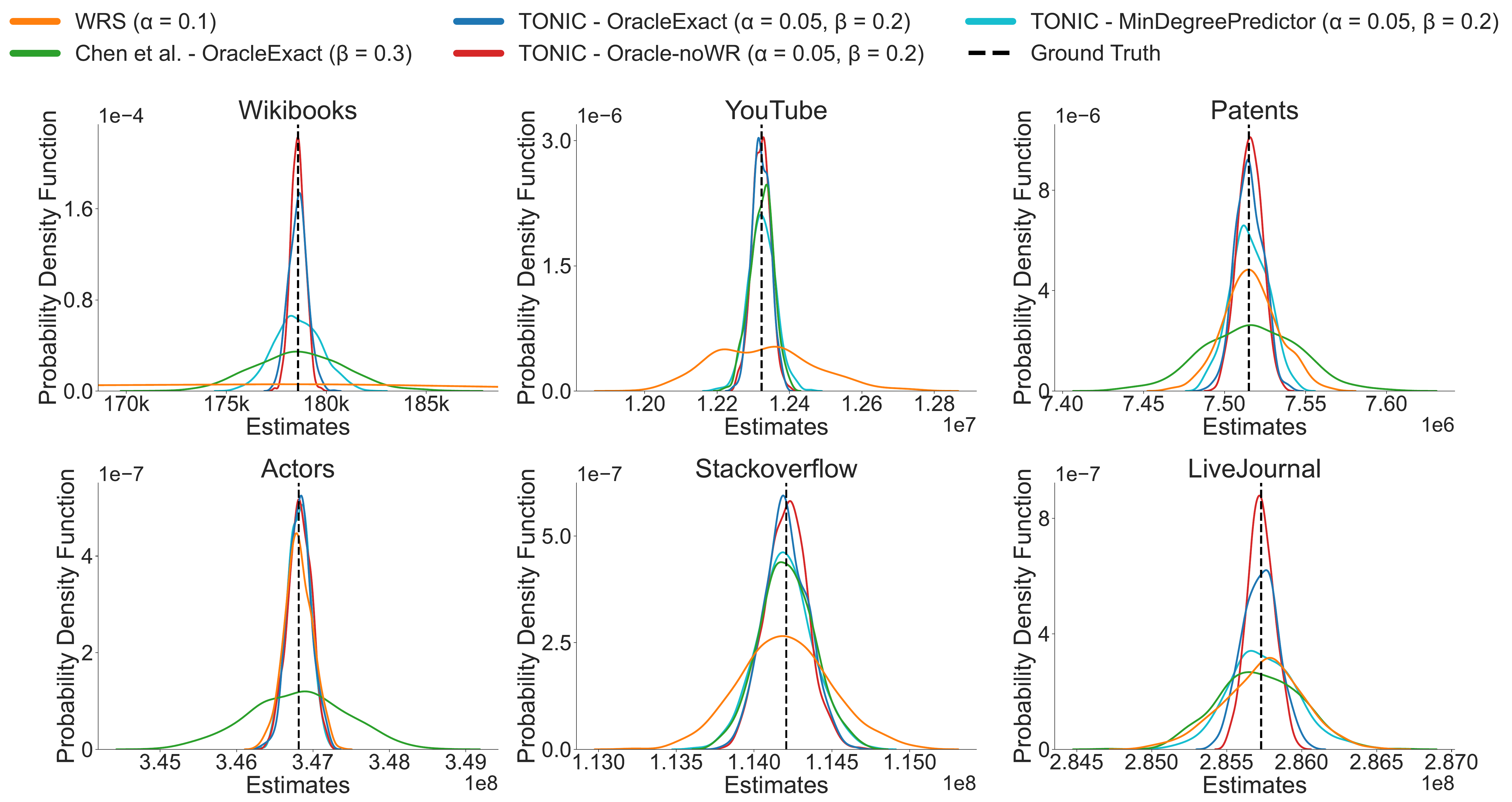}
	\caption{Probability Density Distribution of estimations. For each combination of algorithm and parameter (including oracle for \algname), the distribution of estimates over 500 independent repetitions is shown. The algorithms parameters are as in legend (for \wrs\ and \chenalg\ they are fixed as in the respective publications; for \algname\ they are as chosen in Fig.~\ref{fig:accuracy_params_experiments_merged})}
	\label{fig:pdf_merged-appendix}
\end{figure}

Moreover, we observe that, for each dataset and predictor, including the practical \texttt{MinDegreePredictor}, \algname\ is able to achieve the smallest variance with respect to \wrs\ and \chenalg\ (excluding YouTube, where \chenalg\ has slightly smaller variance than \algname\ empowered by \texttt{MinDegreePredictor}).

\subsection{Accuracy at any time}
\label{sec:stream_experiments}

In this section, we study the behavior of \algname\ in the terms of quality of global triangles estimates at any time $t$, during the evolving of the input graph stream. 
We considered only \wrs\, since \chenalg\ is not suited to return estimates at a given time. 
Let $t = 1, 2, ..., m$ be the time of arrival of the $t$-th edge in the input stream $\Sigma$. We quantized the time interval in order to obtain around 1k global estimates, hence we collected outputs at times multiple of $\tau$, where $\tau \approx m / 1000$. 
Results are displayed in Fig.~\ref{fig:stream_experiments_merged-appendix}, where mean and standard deviation are reported. Notice that \algname\ is able to significantly outperform \wrs\ for each dataset, for each predictor, and for the all duration of the graph stream, except for the last chunk of stream of Twitter-merged. Also notice that, since here we fixed $k = m/10$ (indicated with an arrow), the error is zero before reaching such number of incoming edges, since we count triangles inside our sample with probability 1, as expected. 

\begin{figure}[htbp]
	\centering
	\includegraphics[width=1 \textwidth]{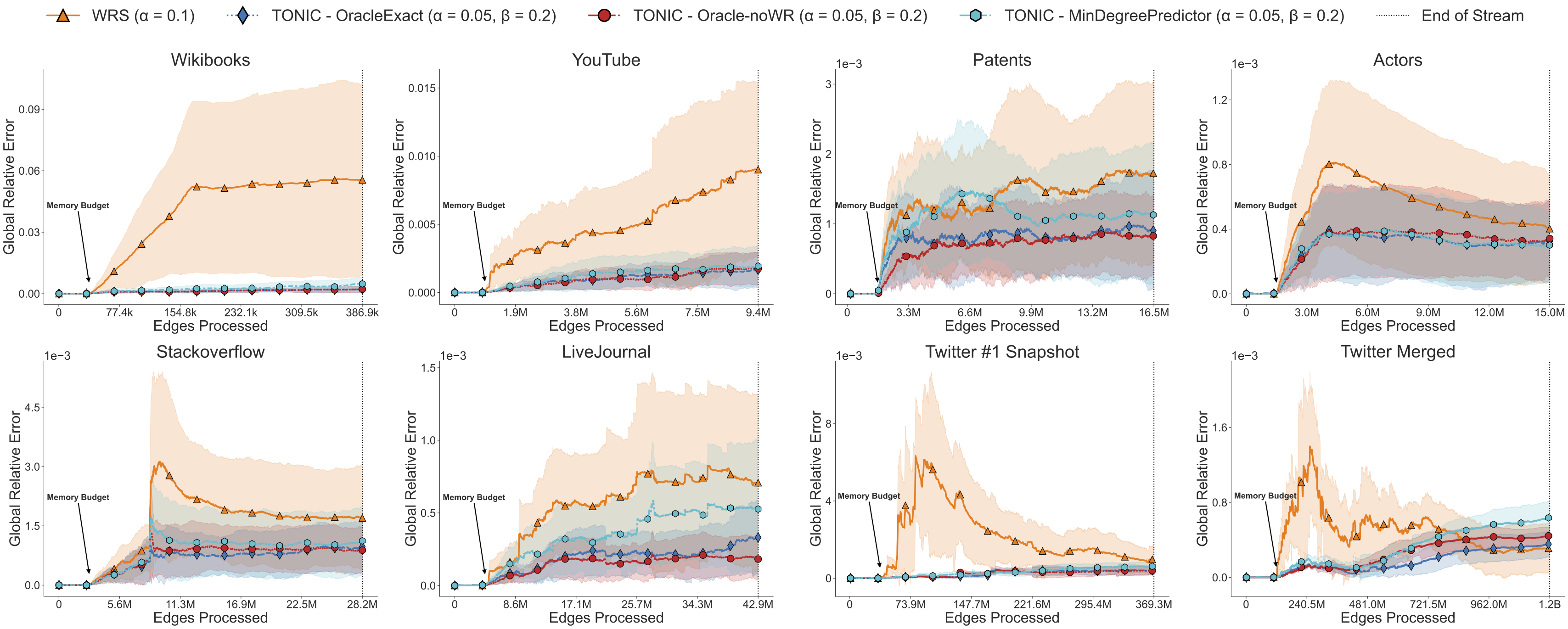}
	\caption{Estimation error as time progresses during the input graph stream. For each combination of algorithm and parameter (including predictor for \algname), the average and standard deviation over 50 repetitions are shown (except for Twitter datasets, where 10 repetitions have been considered). The algorithms parameters are as in legend (for \wrs\ they are fixed as in the respective publications; for \algname\ they are as chosen in Fig.~\ref{fig:accuracy_params_experiments_merged})}
	\label{fig:stream_experiments_merged-appendix}
\end{figure}

\subsection{Number of Discovered Triangles}
\label{sec:discovered _triangles}

In Fig.~\ref{fig:triangles_distributions_merged-appendix}, we display the number of counted and estimated triangles for the considered datasets. 
Notice that \algname\ (also the version empowered by \texttt{MinDegreePredictor}) is able, in almost any case, to count more triangles with at least 1 edge in the waiting room and/or in the heavy edge set (i.e., edges that lead to smaller variance), since it takes advantage of features both from the waiting room and from the heavy edges predictions. Also, note that, in general \algname\ corrects the number of counted triangles more accurately  than \chenalg, thanks to the reservoir sampling strategy. Again, we recall that the space for storing heavy edges for \chenalg\ is $0.3 k$, while the space for storing heavy edges in \algname\ is $0.19k$. 

\begin{figure}[htbp]
	\centering
	\includegraphics[width=1 \textwidth]{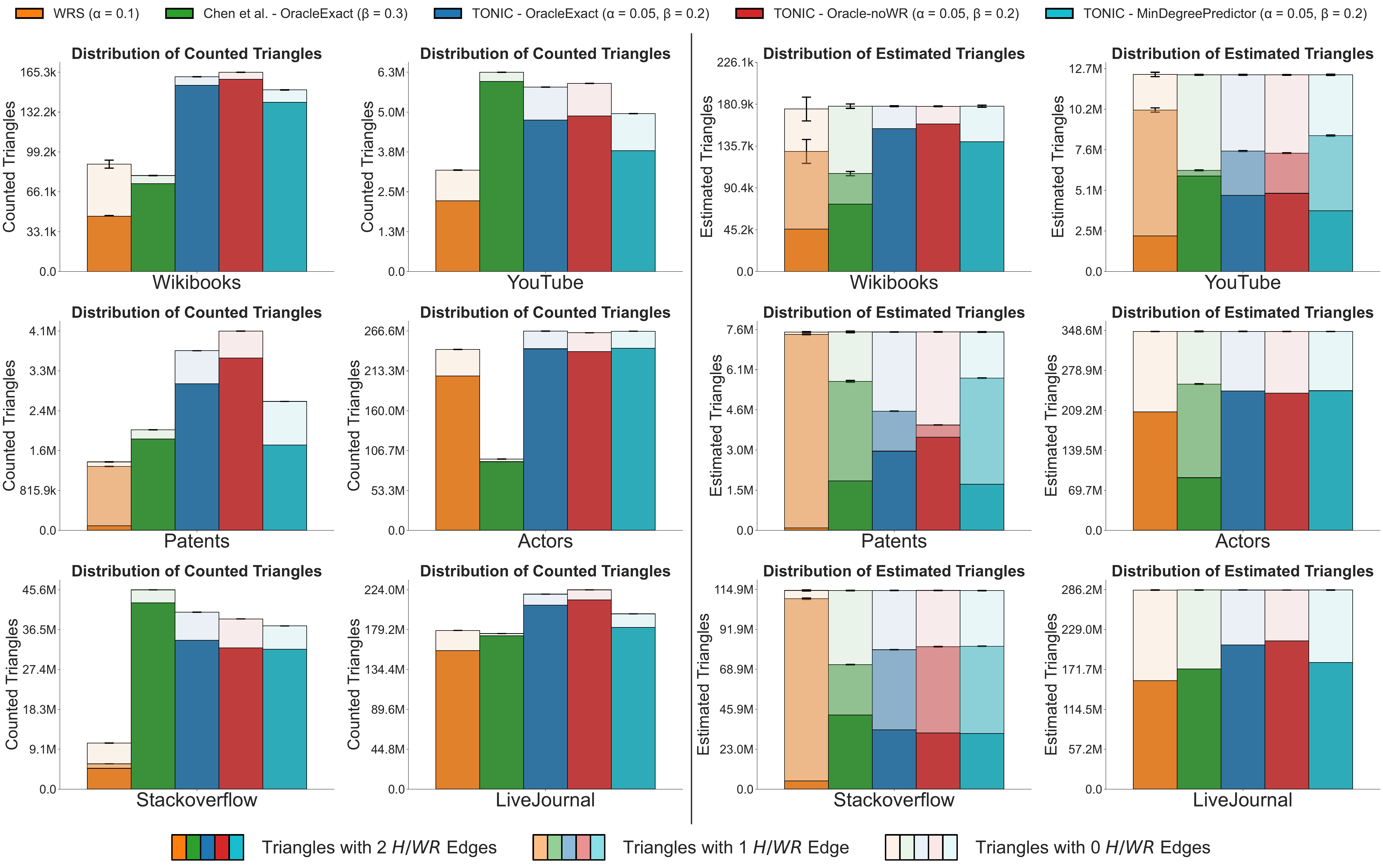}
	\caption{Distributions of Counted (left) and Estimated (right) Triangles for each dataset. For each combination of algorithm and parameter (including oracle for \algname), the average and standard deviation over 50 repetitions are shown. The algorithms parameters are as in legend (for \wrs\ and \chenalg\ they are fixed as in the respective publications; for \algname\ they are as chosen in Fig.~\ref{fig:accuracy_params_experiments_merged})}
	\label{fig:triangles_distributions_merged-appendix}
\end{figure}

\subsection{Local Triangles Experiments}
\label{sec:local_triangles_experiments}
We also performed experiments to assess \algname\ in estimating local triangle counts.
We recall that, by definition, given a node $u \in \mathcal{V}$, the number of local triangles $T_u$ associated to node $u$ is the number of triangles in which $u$ is involved. Note that $\sum_{u \in \mathcal{V}} T_u = 3T$. 
We considered the top 20\% nodes with highest local triangles count, denoted as $V_{top_{20}}$. Given the local triangles estimate $\hat{T}_u$, we computed the following metrics:
\begin{myitemize}
	\item \emph{Local Relative Error}, defined as:
	\begin{small}
	\begin{equation*}
		\frac{1}{V_{top_{20}}} \ \sum_{u \in V_{top_{20}}} \left| \hat{T}_u - T \right| / \ T_u \text{ (the lower the better);}
	\end{equation*}
	\end{small}
	\item \emph{Spearman Rank Coefficient}~\cite{spearman1961proof} calculated between local triangles ground truth in $V_{top_{20}}$, and the corresponding rank of such nodes in local triangles estimates (the higher the better, max value = 1).
\end{myitemize}
In Fig.~\ref{fig:local_triangles-appendix}, we show the results averaged over 10 independent trials, reporting mean and standard deviation. 
We provided all the algorithms with memory budget $k = fm$, with ~$m=|E|$, for $f=0.01, 0.03, 0.05, 0.1$. Note that, for some dataset, \algname\ is comparable or slightly worse than \wrs. This is due to the fact that in \algname\ the predictor focuses on heaviest edges, that account for more precise global estimates, at the expense of the estimated count of the local triangles, in particular for vertices with lower local counts.

\begin{figure}[htbp]
	\centering
	\includegraphics[width=1 \textwidth]{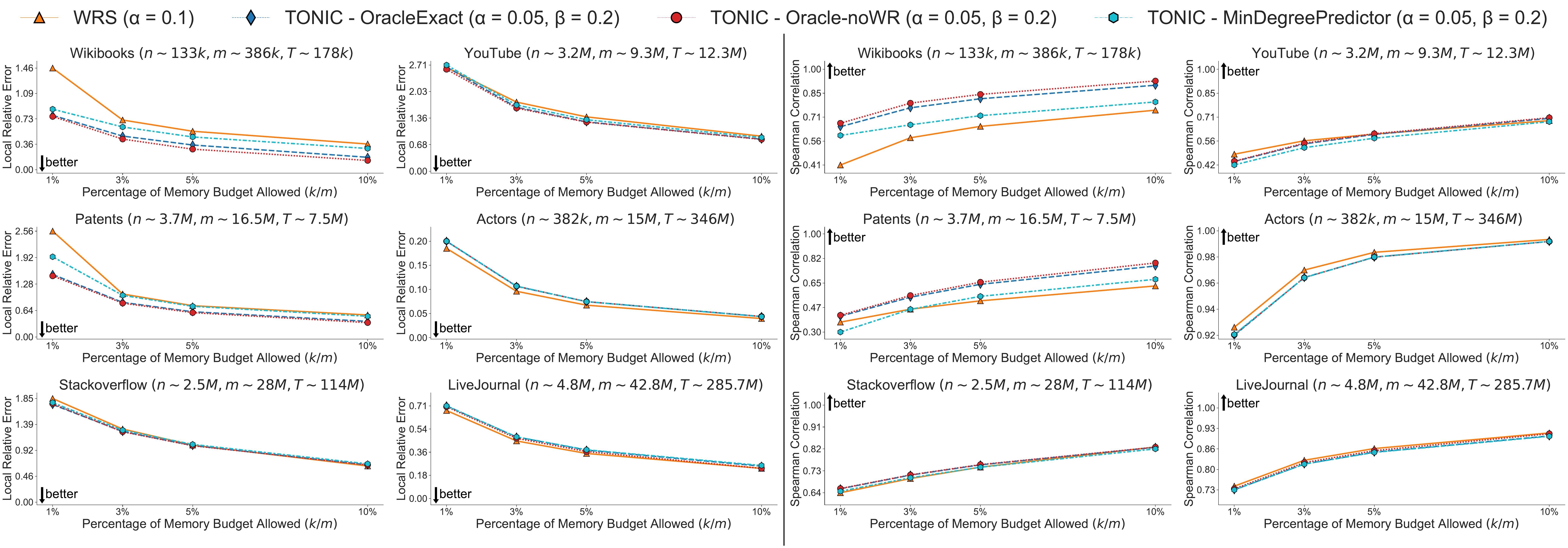}
	\caption{Error (left) and Spearman Rank Coefficient (right) vs memory budget. For each combination of algorithm and parameter (including predictor for \algname), the average and standard deviation over 10 repetitions are shown. The algorithms parameters are as in legend (for \wrs\ they are fixed as in the respective publications; for \algname\ they are as chosen in Fig.~\ref{fig:accuracy_params_experiments_merged})}
	\label{fig:local_triangles-appendix}
\end{figure}

\subsection{Experiments with Twitter Single Graphs}
\label{sec:scalability-twitter-snaps}
In this section we provide results also for the second, third and fourth snapshot of Twitter network (first snapshot and the merged version of Twitter have been already discussed in Fig~\ref{fig:memory_experiments_merged}). In Fig.~\ref{fig:twitter-snapshots-experiments}, we show results for accuracy and runtime by varying the memory budget allowed to the algorithms. Again, Fig.~\ref{fig:twitter-snapshots-experiments} (right) shows the runtime only for \algname\ and \wrs, since the runtime of \chenalg\ is always at least 6 and up to 12 times larger than \algname\ runtime. We see that \algname\ is always faster than \wrs\ (excluding $1\%m$ memory budget for Twitter~\#2 snapshot, for which the runtimes are comparable, while \algname\ is highly outperforming \wrs). Such results confirm that \algname, in practice, is able to scale better with respect to worst-case analyses of the running time. In some cases, \chenalg\ has the best accuracy, slightly better than \algname. This is due to Twitter streams being in adjacency list order, since in such cases the waiting room is not beneficial. In particular, if the current node in the adjacency list stream has a degree comparable to the size of the waiting room, the temporal localities (see Section~\ref{sec:wrs}) are completely lost. However, while the accuracy of \algname\ is worse than, but comparable to, \chenalg, the difference in runtime is huge: for example, in Twitter \#2, for a single run with $k = 10\% m$, \algname\ takes~$\sim 15$ minutes, \chenalg\ requires more than $3$ hours.

\subsection{Adversarial predictors}
\label{sec:adversarial-experiments}
In this section, we focus on the results of \algname\ when it is provided with an extremely bad predictor in snapshot sequences. More specifically, we build \emph{the most adversarial predictor}, that is we set the lightest edges (resp. lowest degree nodes) to be the heaviest (resp. highest degree nodes), and preserving the sizes of the predictors (i.e., the number of entries in the predictor). 
In Fig.~\ref{fig:adversarial_oracles} we report the results containing both best and adversarial predictors. 

\begin{figure}[tbp]
	\centering
	\includegraphics[width=0.9 \textwidth]{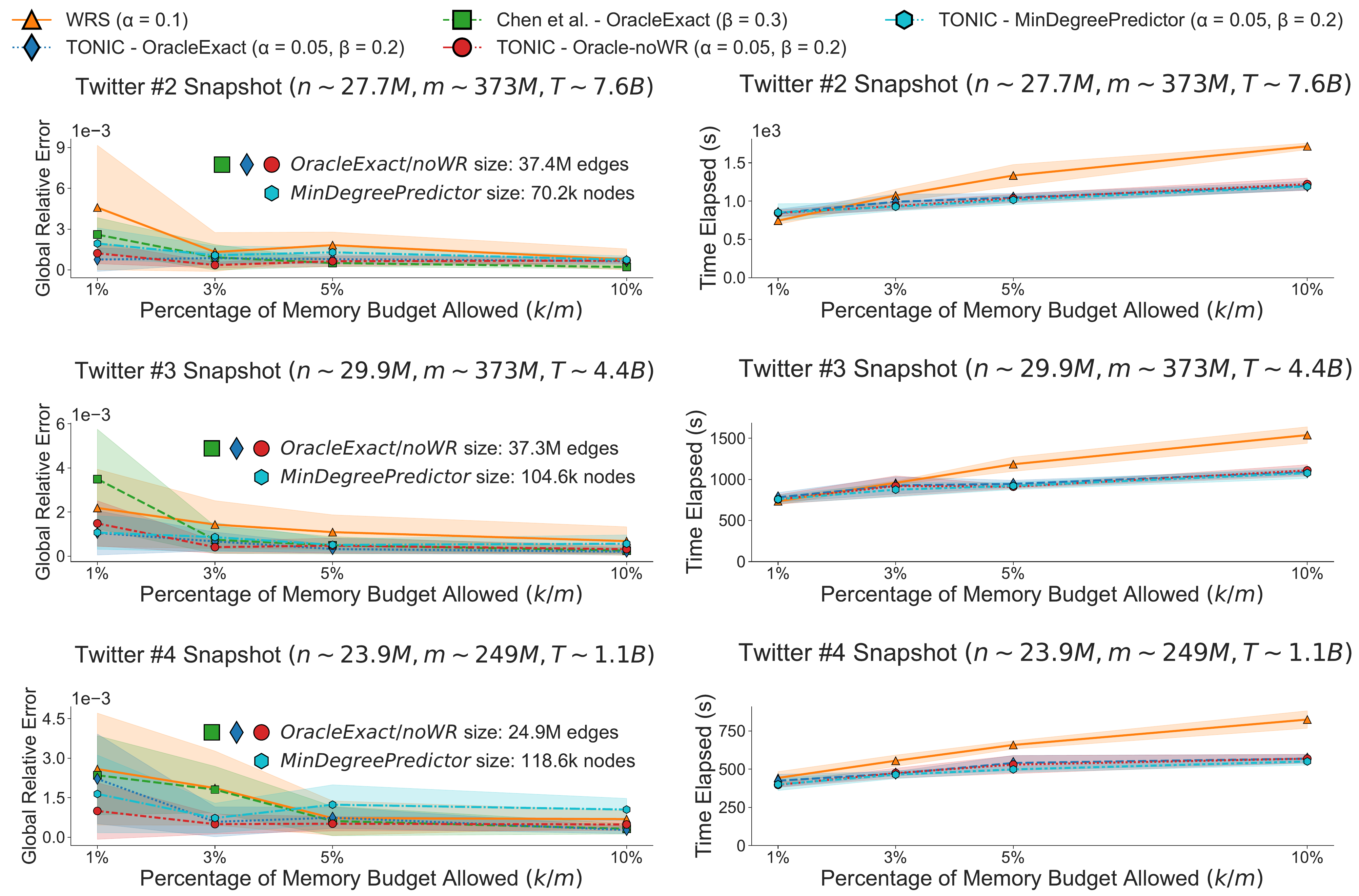}
	\caption{Error (left) and runtime (right) vs memory budget. \chenalg\ runtimes are not shown for clarity, since they are 6-12 times bigger than \algname\ runtime. For each combination of algorithm and parameter (including predictor for \algname), the average and standard deviation over 10 repetitions are shown. The  algorithms parameters are as in legend (for \wrs\ they are fixed as in the respective publication; for \algname\ they are as chosen in Fig.~\ref{fig:accuracy_params_experiments_merged}).}
	\label{fig:twitter-snapshots-experiments}
\end{figure}

Again, we want to emphasize that our node-degree predictor is taking only some dozens of the node-degrees (the highest degrees in the graph for the best \texttt{MinDegPredictor}, and the lowest ones for the adversarial \texttt{MinDegPredictor}). In the specific, we are only keeping 67 out of 10k, 82 out of 16k nodes, and 34 out of 3k respectively for Oregon, AS-CAIDA and AS-733 training graph, i.e., the first graph in the sequence, as reported in the caption of each subplot.
We note that \algname, even with its adversarial components, is almost always outperforming \chenalg\ with the \textit{exact oracle} in AS-CAIDA dataset. In any case, \chenalg\ with adversarial oracle is performing absolutely poor with respect to other methods (while \algname\ with the adversarial oracle is still competitive). Furthermore, \algname\ with adversarial components shows comparable to \wrs\ in AS-CAIDA, AS-733 and for more of the half of the sequence in Oregon. In summary, Fig.~\ref{fig:adversarial_oracles}  shows that \algname\ is robust to poor information of adversarial oracles/predictors, thanks to its combination with waiting room sampling, while is able to take advantage of the information provided when they are useful.

\begin{figure}[htbp]
	\centering
	\includegraphics[width=1 \textwidth]{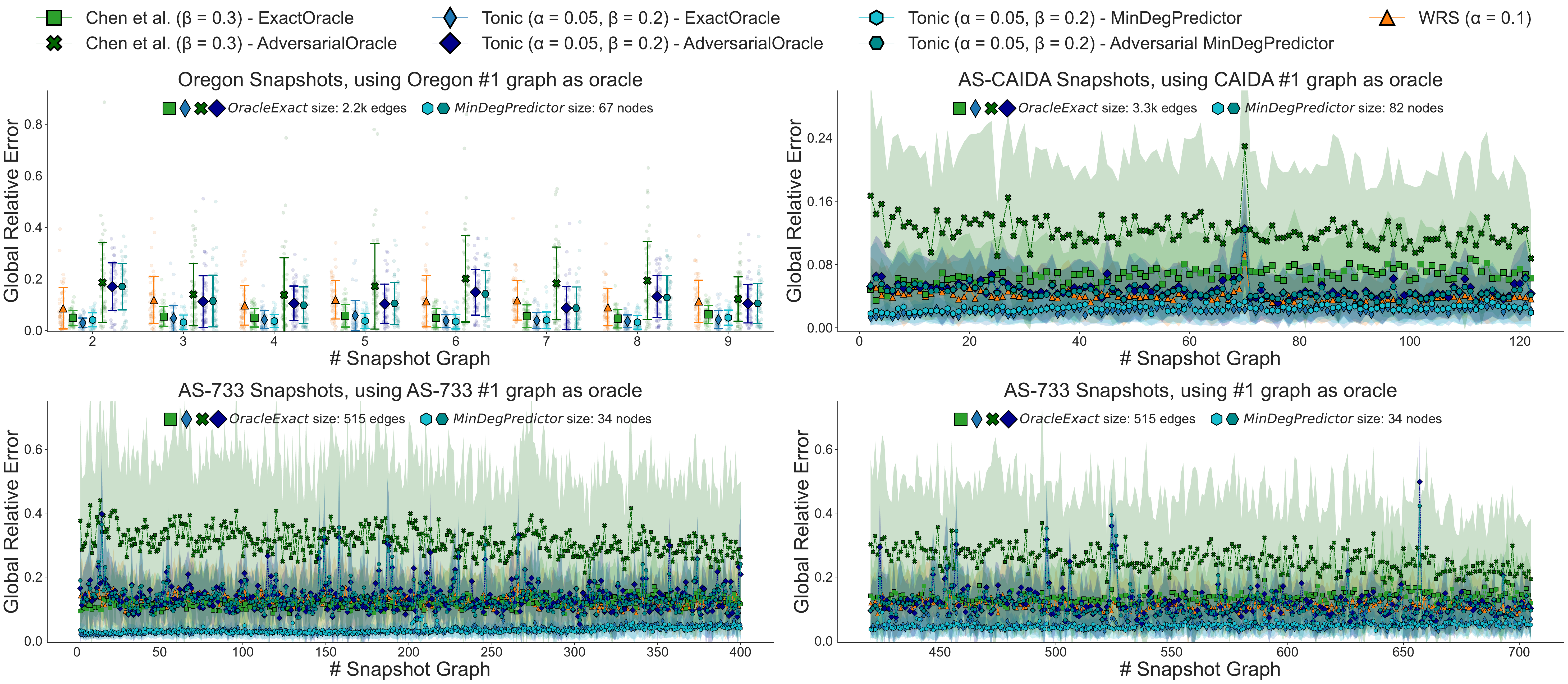}
	\caption{Error with snapshot networks with sequence of graph streams. In all cases the predictors are trained only on the first graph stream of the sequence (with results not shown on such graph stream). For each combination of algorithm and parameter (including predictor for \algname), the average and standard deviation over 50 repetitions are shown. The  algorithms parameters are as in legend (for \wrs\ and \chenalg\ they are fixed as in the respective publications; for \algname\ they are as chosen in Fig.~\ref{fig:accuracy_params_experiments_merged}). }
	\label{fig:adversarial_oracles}
\end{figure}

\subsection{Fully Dynamic Streams Experiments}
\label{sec:fd-experiments}
In the following, we provide a complete experimental evaluation when dealing with fully dynamic (FD) streams, i.e., with graph streams that allow both insertions and deletions of edges. More specifically, we create FD streams considering our 4 graph sequences datasets (see Table~\ref{tab:datasets}): Oregon, AS-CAIDA, AS-733 and Twitter. For each dataset, starting from the first graph of the sequence, we compute edge additions and removals between the current and the next snapshot, for all the snapshots in the sequence. Edge insertions are added to the FD stream by preserving the temporal ordering following the graph sequence, and edge deletions are added to the FD stream with random timestamps inside the time window of the snapshots that we are considering. 
In the following, we refer to our algorithm for FD streams as \algnamedel\, for which we described the general workflow in Algorithm~\ref{alg:tonic_fd}.

To assess performance of \algnamedel\, we consider two groups of experiments: 
(i) error during the evolving of the FD stream,
(ii) error and runtime vs memory budget.
Predictors for \algnamedel\ are trained \textit{only} on the \textit{first} snapshot of the sequence, as described for snapshot experiments in Section~\ref{sec:experiments}, and thus completely not aware of the consequent edge deletion.

\begin{figure}[htbp]
	\centering
	\includegraphics[width=1 \textwidth]{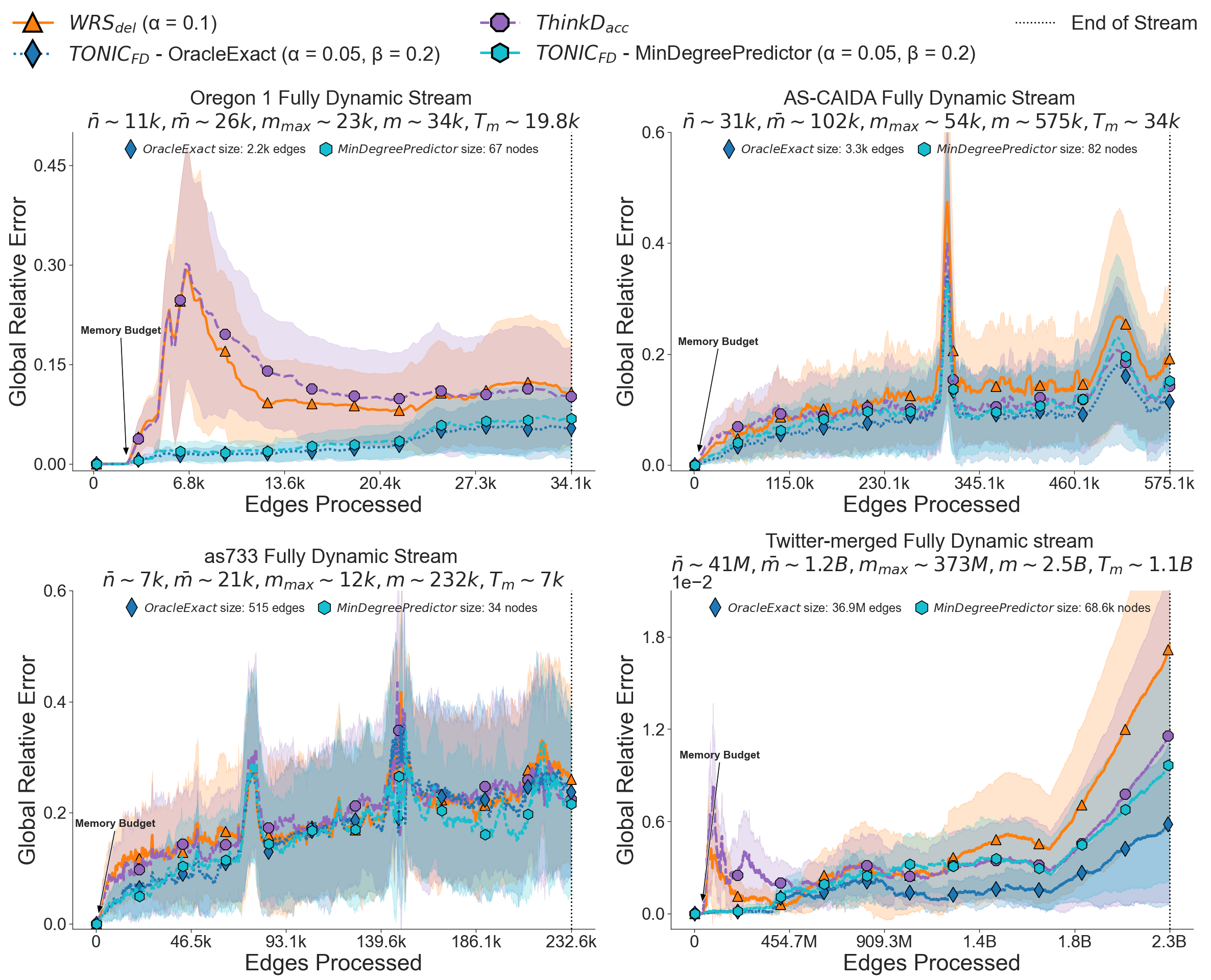}
	\caption{Estimation error as time progresses during the input FD graph stream. For each combination of algorithm and parameter (including predictor for \algnamedel), the average and standard deviation over 50 repetitions are shown (except for Twitter, where 10 repetitions have been considered). The algorithms parameters are as in legend (for \wrsdel\ they are fixed as in the respective publications; for \algnamedel\ they are as chosen in Fig. ~\ref{fig:accuracy_params_experiments_merged}). $\bar{n}$: number of unique nodes;  $\bar{m}$: number of unique edges;  $m_{max}$: maximum number of edges at some time; $m$: total number of edges; $T_m$: number of global triangles at the end, derived from the FD stream.} 
	\label{fig:fd-stream_experiments}
\end{figure}
\begin{figure}[htbp]
	\centering
	\includegraphics[width=1 \textwidth]{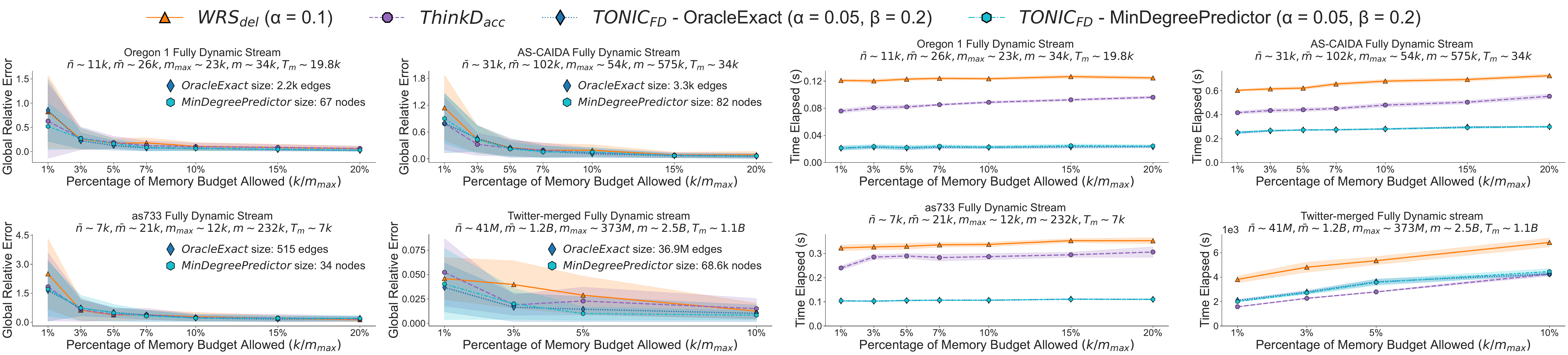}
	\caption{Error (left) and runtime (right) vs memory budget. For each combination of algorithm and parameter (including predictor for \algnamedel), the average and standard deviation over 50 repetitions are shown  (except for Twitter, where 10 repetitions have been considered). The  algorithms parameters are as in legend (for \wrsdel\ they are fixed as in the respective publications; for \algnamedel\ they are as chosen in Fig.~\ref{fig:accuracy_params_experiments_merged}). $\bar{n}$: number of unique nodes;  $\bar{m}$: number of unique edges;  $m_{max}$: maximum number of edges; $m$: total number of edges; $T_m$: number of global triangles at the end, derived from the FD stream.}
	\label{fig:fd-memory_experiments}
\end{figure}

Fig.~\ref{fig:fd-stream_experiments} shows the estimation error as time progresses during the input FD stream. Estimates are obtained as in Section~\ref{sec:stream_experiments}. Notice that \algnamedel\ performs the best in terms of error in almost any case, for all the considered datasets, predictors and all the duration of the FD stream. All the algorithms have been run with memory budget $k = m_{max} / 10$  (highlighted by the black arrow). 

In Fig.~\ref{fig:fd-memory_experiments} we report results of error and runtimes; the three algorithms are provided with memory budget $k~=~f m_{max}$, for values of $f$ showed on the x-axis. 
Note that \algnamedel\ performs always better than \wrsdel, and better or slightly worse than \thinkdacc, in terms of error (Fig.~\ref{fig:fd-memory_experiments}, left). 
Fig.~\ref{fig:fd-memory_experiments} (right) shows that our algorithm \algnamedel\ is always faster than \wrsdel, and is faster than or comparable to \thinkdacc, despite requiring the removal of edges within our waiting room and heavy edge set.

\section{Details on Predictors}
\label{sec:predictors_details}


\subsection{Overhead of node-based oracles}
\label{sec:overhead-node-oracle}

In this section, we provide further details on how we build \texttt{MinDegPredictor}. More specifically, given the training stream, we sort edges by decreasing minimum degree, and we retain the top 10\% of such edges. Thus, each entry of the predictor (represented as a lookup table) can be written as $((u, v) ; min(deg(u), deg(v)))$, for the $m / 10$ highest-min-degree entries. Then, given the number $\bar{n}$ of (unique) nodes present in the above \emph{edge-based} predictor, we computed \emph{node-based} \texttt{MinDegPredictor}, where the latter contains the $\bar{n}$ highest degree nodes of the whole graph stream.
We want to emphasize the fact that nodes in the edge-based predictor do not necessarily correspond to the same nodes taken from the highest-degree nodes in the graph: for example, think about a very high-degree node that has all low-degrees neighbors, i.e., it would not be among the unique nodes in the edge-based predictor, but it is among the highest-degrees of the graph, thus inside the node-based predictor.
Node-based predictors are in someway encoding edge features in a much more succinct representation. Hence, the resulting predictors' overhead is significantly lower than the one used in edge-based representations. Note that in our node-based predictor we combine the degrees to obtain the minimum, but many others predictors based on the degree of nodes could be used (e.g., a linear combination with trainable parameters). We leave the exploration of such predictors as future direction.  In Table~\ref{tab:overhead}, for each row (dataset) we report the following factors: gain in the number of entries, gain from storing in memory, gain from reading from file, comparing node-based and edge-based representation. 

\begin{table}[htbp]
	\centering
	\footnotesize
	\resizebox{1 \textwidth}{!}{
		\begin{tabular}{c c c c c}
			\textbf{Dataset} 						& \textbf{Gain in \# of Entries} & \textbf{Gain in Memory} & \textbf{Gain in Time} &  \textbf{Jaccard Similarity} \\
			\toprule
			Patents    									  & $7.90$			& $13.04$		 &    $9.4 \pm 0.3$	  		              & $0.9600$	 				\\
			Actors				   						   & $189.11$		&  $288.55$	   &    $65.0 \pm 15.2$			          & $1.0$	 				  \\
			Stackoverflow    						 & $147.93$		 &  $232.72$    &    $72.7 \pm 16.0$		            & $0.9998$	 			  \\
			Twitter \#1 snapshot   				 & $538.44$ 	&  $879.57$    &    $334.4 \pm 188.2$	            & $0.9998$	  		           \\
			Twitter-merged 					   & $734.96$ 	  &  $1201.52$   &	  $432.9 \pm 284.87$			& $0.9990$	  	 	        \\
			\bottomrule
		\end{tabular}
	}
	\caption{Statistics on node-based \texttt{MinDegreePredictor} vs edge-based \texttt{MinDegreePredictor}. Second, third and fourth columns represent respectively the ratio of the number of entries, the size in memory, and the time for reading from file of node-based predictor over edge-based predictor (gain factors). 
		The last column indicates the Jaccard Similarity.}
	\label{tab:overhead}
\end{table}

We notice that we are able to obtain huge savings in terms of space to store and time to read the predictor. To give additional numbers: for Twitter-merged (last row) we have edge-based representation size in memory of $\approx 2.7 GB$, and oracle's time to be loaded of $\approx 60 s$, while for node-based representation we have respectively $ \approx 2.2MB$ and $\approx 0.12 s$. Then, we computed the Jaccard Similarity (last column) between the set of unique nodes in edge-based predictor and the set of nodes in node-based predictor, that is computed as the ratio between the cardinality of the intersection and the cardinality of the union of such two sets. Jaccard Similarity provides a measure for gauging similarity and diversity of sets (highest value is 1, corresponding to identical sets). From last column on Table~\ref{tab:overhead}, we see that nodes contained in the two different representations are basically the same. We recall that sizes of predictors are in each subplot of Fig.~\ref{fig:memory_experiments_merged} (left).


\subsection{Time to build predictors/oracles}
\label{sec:time-to-build-oracles}
In the following, we present results for comparison between times for building \texttt{OracleExact} and \emph{node-based} \texttt{MinDegPredictor}. We recall that the former requires to solve exactly the problem of counting the number of triangles in the graph, and plus maintaining a lookup table with the value of the heaviness for each edge $e$, i.e., the number of triangles adjacent to edge $e$; the latter, is built with a fast pass on the stream and stores the degrees of the nodes in the lookup table. In the end, the edge table (resp. node table) is sorted by heaviness (resp. degree) and the top heaviest edges (resp. highest degrees nodes) are retained. In Table~\ref{tab:timetobuildoracles} we report times for building oracles and predictors for some representative datasets, averaged over 5 independent runs. Note that \algname\ is able to outperform or is comparable to \wrs\ in runtime even if we add times from Table~\ref{tab:timetobuildoracles} to algorithms' execution times in Fig.~\ref{fig:memory_experiments_merged} (right), for most of the combinations of datasets and memory budgets, i.e., if we consider the time for both the first pass for building \texttt{MinDegPredictor}, and the second pass to run \algname. This confirms the time practicability and efficiency of our proposed predictor and algorithm.

\begin{table}[htbp]
	\centering
	\footnotesize
	\resizebox{0.7 \textwidth}{!}{
		\begin{tabular}{c c c}
		    Dataset 						& ${t_{\texttt{OracleExact}} (s)}$ 					& ${t_{\texttt{MinDegreePredictor}} (s)}$ \\
			\toprule
			Patents    									  &  $44.61 \pm 3.15$											  		& $7.78 \pm 0.35$				 \\
			Actors				   						   &  $158.22 \pm 7.68$										   		  & $5.68 \pm 0.42$	 				 \\
			Stackoverflow    						 &  $792.41 \pm 28.59$										       & $11.27 \pm 0.12$	 			   \\
			Twitter \#1 snapshot   				 & $13.35 \times 10^3 \pm 1117.14$ 	 					     &  $178.47 \pm 4.75$	  		  \\
			Twitter - merged 					   &	$64.35 \times 10^3 \pm 6206.57$ 					   & $561.63 \pm 16.42$	  		\\
			\bottomrule
		\end{tabular}
	}
	\caption{Time to build \texttt{OracleExact} vs Time to build \texttt{MinDegreePredictor}. The results include also the times for sorting the lookup tables, retaining the top entries and writing to file. For each dataset, average and standard deviation over 5 independent repetitions are reported.}
	\label{tab:timetobuildoracles}
\end{table}

\end{document}